\documentclass[12pt,draftclsnofoot,onecolumn]{IEEEtran}
\usepackage{enumitem}
\usepackage{amsthm}
\newtheorem{lemma}{Lemma}
\newtheorem{theorem}{Theorem}
\usepackage{float}
\usepackage{algorithm}
\usepackage{algorithmicx}
\usepackage{algpseudocode}
\usepackage{stfloats}
\ifCLASSINFOpdf
\else
\fi

\usepackage{CJK}
\usepackage{graphicx}
\usepackage{caption}
\usepackage{epstopdf}
\usepackage{subfigure}
\usepackage[cmex10]{amsmath}
\usepackage{amssymb}
\usepackage{mathabx}
\usepackage{multirow}
\usepackage{xcolor}
\usepackage{multirow}
\usepackage{bm}
\usepackage{mathabx}
\usepackage{stmaryrd}
\usepackage{upgreek}
\begin{document}

\title{Performance Analysis and Optimization of Cooperative Full-Duplex D2D Communication Underlaying Cellular Networks}

\author{Guoling~Liu,~Wenjiang~Feng,~Zhu~Han,~\IEEEmembership{Fellow,~IEEE},\\and~Weiheng~Jiang,~\IEEEmembership{Member,~IEEE}
\thanks{G. Liu, W. Feng and W. Jiang are with the College of Communication Engineering, Chongqing University,
Chongqing, 400044 China~(e-mail: \{liuguoling, fengwj, whjiang\}@cqu.edu.cn).}
\thanks{Z. Han is with the Department of Electrical and Computer Engineering, University of Houston, Houston, TX 77004 USA, and also with the Department of Computer Science, University of Houston, Houston, TX 77004 USA (e-mail: zhan2@uh.edu).}}

\maketitle

\begin{abstract}
This paper investigates the cooperative full-duplex device-to-device (D2D) communication underlaying a cellular network, where the cellular user (CU) acts as a full-duplex relay to assist the D2D communication. To simultaneously support D2D relaying and uplink transmission, superposition coding and successive interference cancellation are adopted at the CU and the D2D receiver, respectively. The achievable rate region and joint outage probability are derived to characterize the performance of the considered system. An optimal power allocation scheme is proposed to maximize the minimum achievable rate. Besides, by analyzing the upper bound of the joint outage probability, we study a suboptimal power allocation to improve the outage performance. The simulation results confirm the theoretical analysis and the advantages of the proposed power allocation schemes.
\end{abstract}


\IEEEpeerreviewmaketitle

\section{Introduction}
To meet the booming data demand of emerging wireless communication services, researchers in academia and industry are seeking for new technologies to reform the traditional cellular networks. As an attractive candidate, device-to-device (D2D) communication \cite{asadi_survey} underlaying cellular networks draws wide attention. By allowing direct communication between proximal users without traversing the base station (BS) or core network, D2D communication provides improvement in spectrum efficiency, energy efficiency and communication delay, and thus enables high-rate proximity-aware services, such as media sharing, social network and gaming \cite{song2015wireless}. Consequently, D2D communication is considered as a promising technology in the next generation wireless network \cite{tehrani2014device}.
\subsection{Background of Cooperative D2D and Full-Duplex Communication}
Although D2D communication can support high spectrum efficiency, the actual transmission rate of D2D users is restricted by practical constraints, such as modulation and coding schemes. Therefore, the channel capacity of D2D links is insufficiently utilized \cite{tang2017cooperative}. To address the redundant capacity problem, cooperation is introduced to D2D communication for coverage extension and performance enhancement of the cellular networks. Reference \cite{shalmashi2014cooperative} allows a D2D transmitter (DT) to act as a relay to assist the downlink cellular transmission and at the same time transmit its own data to a D2D receiver (DR) by employing superposition coding (SC) \cite{popovski2008improving}. The cellular user performs maximum-ratio-combining to decode the data from the BS and DT. And the DR uses successive interference cancellation (SIC) \cite{blomer2009transmission} to decode the D2D data. Therefore, both the cellular and D2D users can benefit from cooperation. The same cooperation scheme is also investigated in the cellular uplink and overlay D2D scenario \cite{cao2015cooperative}, where the employment of SC is replaced by orthogonal radio source allocation for simultaneous cellular and D2D communication. In \cite{pei2013resource}, the authors assign the DT as a two-way relay to establish a bidirectional cellular link while communicating with a DR. A relay selection method is proposed to achieve a larger rate region. The limited battery lifetime of D2D users is considered in \cite{wang2012icc, wang2012energy}. The authors use auction game to modelled the resource allocation problems, and proposed corresponding auction algorithms to optimized the energy efficiency.

The aforementioned researches focus on the half-duplex (HD) cooperation, which suffers a loss of spectrum efficiency. Recently, the in-band full-duplex (FD) cooperation \cite{song2017full,zhang2015full,song2015resource} is a frequent topic. With the breakthrough of self-interference suppression (SIS), simultaneous transmission and reception in the same frequency band becomes practical in realistic wireless networks. The deployment of FD communication may overcome the loss of spectrum efficiency due to HD cooperation. However, the main shortcoming of FD communication is the residual self-interference (RSI) \cite{duarte2012experiment} caused by imperfect SIS. The feasibility and superiority of FD communication in non-cooperative D2D networks has been demonstrated in \cite{ali2015effect,chai2016throughput,tang2016energy}. In \cite{zhang2017full}, the FD DT cooperates with the BS to perform non-orthogonal multiple access (NOMA) \cite{ding2015cooperative} and improves the outage performance of the user with weak cellular downlink. An adaptive multiple access switching method is proposed to dynamically choose the optimal multiple access scheme. As a dual-hop version of the model in \cite{zhang2017full}, Zhang \textit{et al.} investigate the optimal power allocation to minimize the outage probability \cite{zhang2017performance}. To address the fairness issue between the NOMA-strong and NOMA-weak user, another power allocation scheme is studied to maximize the minimum rate achieved by the cellular and D2D link. Reference \cite{zhang2015power} considers the same cooperative scheme as in \cite{shalmashi2014cooperative} with the DT operating in the FD mode. The cellular and D2D data are superposed in different power levels at the DT. Under the aggregate power constraint, an optimal power allocation algorithm is designed to maximize the achievable rates for the D2D users while fulfilling the minimum rate requirements of the cellular users. However, there is no SIC employed at the receiver of the cellular user and DR to deal with the mutual interference, which can be unsubstantial in some circumstances. The same model with amplify-and-forward relaying is discussed in \cite{dun2017transmission}. A D2D based multicast service is considered in \cite{zhang2015using}, one FD user equipment (UE) helps the BS convey data to a group of UEs. The FD D2D based multicast protocol has higher power efficiency than existing schemes, but the group size is limited to two UEs.
\subsection{Motivation and Related Work}
All the works in \cite{zhang2017full, zhang2017performance, zhang2015power, dun2017transmission, zhang2015using} consider that there is always a direct link between DT and DR. When the DT and DR are separated far away from each other or the D2D link has poor quality, the D2D users either abandon the transmission or resort to the BS for data relaying \cite{tang2017cooperative}, which limits the advantage of the D2D communication. To this end, relay-aided D2D communication becomes an urgent topic. Under the background of HD relaying, many works, including performance analysis \cite{lin2014stochastic, yang2017transmission}, relay selection \cite{ma2017relay}, mode selection \cite{ma2017mode}, resource allocation \cite{hasan2014resource, hasan2014distributed, kishk2017distributed} and energy saving \cite{al2016relay}, etc., are devoted to the investigation of relay-aided D2D communication. By introducing FD relaying, the performance of relay-aided D2D network can be further improved. In \cite{dang2017outage}, Dang \textit{et al.} design a dual-hop FD relay-assisted D2D scheme underlaying a cellular uplink transmission and propose a suboptimal power allocation scheme to minimize the outage probability of the D2D under aggregate power constraint of the DT and the relay. The quality-of-service (QoS) of the cellular user is provisioned by the power control method at the DT. Subsequently, this work is extended to a multi-user OFDMA scenario \cite{dang2017resource}. The relay selection problem in cooperative D2D networks is considered in \cite{ma2016matching}. A matching theory based relay selection method is proposed to minimize the power consumption of D2D users.

In aforementioned works, the coverage extension and performance improvement of the D2D communication are implemented through an extra relay node (acted by an idle D2D user or a dedicated relay). As shown in the literature, this relay node complicates the interference between the cellular and the D2D link, and requires more sophisticated interference manage technology. As a supplement to existing researches, we aim to ameliorate the relay-aided D2D networks from the following aspects. Firstly, the FD relaying is employed to increase the spectrum efficiency. Secondly, we consider the cooperation between the D2D and the cellular users to cope with the interference caused by data relaying. Thirdly, we jointly optimize the performance of both the cellular and the D2D users by the resource allocation method. 
\subsection{Our Contributions}

Inspired by existing works, we propose a new cooperative D2D scheme where the uplink cellular user (CU) acts as an FD decode-and-forward relay between a pair of D2D users, taking RSI at the CU and power control at the DT into account. The CU employs NOMA to support concurrent uplink and D2D communication. Specifically, the CU superimposes the uplink and D2D data with different power levels, and then broadcasts the superposed data to the BS and the DR. The BS and the DR extract their desired data according to a predefined decoding order. The contributions of this paper are summarized as follows.

\begin{enumerate}[label=\arabic*)]
\item \textit{Performance Analysis}: We present the achievable rate region of the cellular uplink rate versus the dual-hop D2D link rate. The Pareto boundary of the region is founded by jointly optimizing the transmit power and power splitting factor at the CU. Besides, we analyse the exact and asymptotic expressions of the joint outage probability of the cellular and D2D link. 
\item \textit{Power Allocation}: Two power allocation schemes are studied in this paper. In consideration of fairness between the cellular and cooperative D2D communication, a maximization problem of minimum achievable rate is formulated at first. Then, due to the intractability of the exact joint outage probability, we derive its upper bound and formulate a relaxed minimization problem of joint outage probability. Both optimization problems are proved to be quasi-concave and have unique solutions.
\item \textit{Simulation and Discussion}: We use Monte Carlo simulation to validate the correctness of performance analysis and the advantage of the proposed power allocation schemes. In the end, we illustrate the impact of RSI on the network performance and compare with the HD network.
\end{enumerate}
\subsection{Organization}
The rest of this paper is organized as follows. Section \ref{section_ii} presents the system model and fundamental assumptions. Detailed analyses of achievable rate region and joint outage probability are provided in Section \ref{section_iii}. In Section \ref{section_iv}, we analyzed the maximization problem of the minimum achievable rate. The relaxed minimization problem of joint outage probability is investigated in Section \ref{section_v}. Simulation results and discussions are shown in Section \ref{section_vi}. In the end, we conclude this paper in Section \ref{section_vii}.

\begin{figure}[!t]
\centering
\includegraphics[width=7cm]{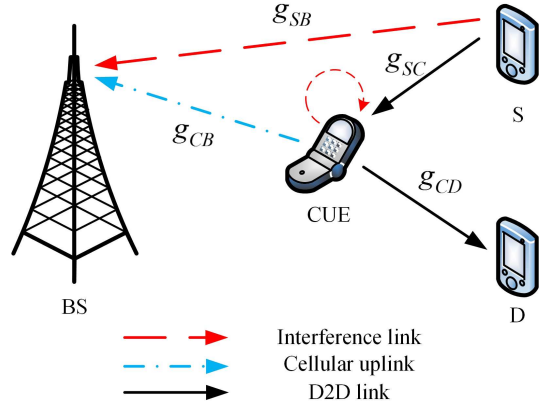}
\caption{System model of cooperative full-duplex D2D communication underlying a cellular network.}
\label{fig_system_model}
\end{figure}

\section{System Model}\label{section_ii}
As shown in Fig. \ref{fig_system_model}, the proposed cooperative FD D2D system consists of one BS, one FD CU, one pair of HD DT and DR, which are denoted as B, C, S and D, respectively. In each transmission period, the CU sends a message to the BS. Meanwhile, the CU acts as an FD decode-and-forward relay to assist the D2D transmission from S to D. To improve the spectral efficiency, the cooperative D2D transmission reuses the cellular uplink channel. Each node is equipped with a single antenna. The channel gain between nodes $i$ and $j$ is denoted as  $g_{ij}$, $i, j \in \{\textrm{B, C, S, D}\}$ . We consider Rayleigh fading, i.e., $g_{ij} \sim \mathcal{CN}\left( 0, \varphi_{ij} \right)$, where $\varphi_{ij}$ is the average power gain of the corresponding channel. The direct link between S and D is ignored due to heavy shadowing or path-loss, i.e., $g_{SD}=0$.

\subsection{Signal Model}
\subsubsection{Power Control Method}
In order to manage the interference at the BS, the truncated channel inverse power control \cite{memmi2017power} is adopted at S. The transmit power of S can be expressed as
\begin{equation}
p_S = \min \left( \frac{\theta}{h_{SB}}, P_S \right),
\end{equation}
where $\theta$ is the maximum tolerable interference threshold predefined by the BS, $P_S$ is the maximum transmit power of S, $h_{ij}=\lvert g_{ij} \rvert^2$ denotes the instantaneous power gain of the channel between nodes $i$ and $j$.

\subsubsection{Statistical Model of Residual Self-Interference}
Following the previous work in \cite{krikidis2013full,rodriguez2014performance,rodriguez2014optimal}, we model the RSI at the CU, $v_C$, as an additive and Gaussian random variable,
\begin{equation}\label{rsi_model}
v_C \sim \mathcal{CN}(0,\beta p_C^\lambda),
\end{equation}
where $p_C \in [0, P_C]$ is the transmit power of the CU, $P_C$ is the maximum transmit power of the CU, $\beta \in [0, +\infty)$ and $\lambda \in [0, 1]$ reflect the performance of SIS. Define the Transmit-power-to-RSI ratio (TRR) as
\begin{equation}\label{trr}
{\rm{TRR}} = \frac{p_C}{\beta p_C^\lambda}.
\end{equation}
Unlike the RSI model in \cite{zhang2017full, zhang2017performance, zhang2015power, dun2017transmission, zhang2015using, dang2017outage,dang2017resource,ma2016matching}, we can see that the TRR is not constant, but an increasing function of $p_C$. The RSI model in (\ref{rsi_model}) incorporates two important cases: (i) $\lambda=0$ indicates a constant RSI level, in this case, the RSI behaves like the noise; (ii) $\lambda=1$, the RSI grows linearly with $p_C$, as in the aforementioned researches. As will be shown in the subsequent section, the value of $\lambda$ has a major influence on the system performance.
\subsubsection{Transmission Protocol}
To facilitate the depiction, we divide the transmission protocol into two concurrent phases.
\begin{itemize}
\item \textit{Phase-I}: In transmission period $t$, S sends a message $x_S(t)$ with power $p_S$ to the CU, the received signal at the CU is
\begin{equation}
y_C \left( t \right) = \sqrt{p_S} g_{SC} x_S \left( t \right) + v_C \left( t \right) + n_C \left( t \right),
\end{equation}
where $n_i (t) \sim \mathcal{CN}(0, \sigma_i^2)$ denotes the additive Gaussian white noise (AWGN) at node $i$.

\item \textit{Phase-II}: After decoding $y_C(t)$, the CU forms a broadcasting signal as
\begin{equation}
\bar{x}_C (t)=\sqrt{\alpha p_C} x_C (t) + \sqrt{(1-\alpha) p_C} \hat{x}_S (t-t_0),
\end{equation}
where $\alpha \in [0,1]$ is the power splitting factor which represents the proportion of the power allocated to $x_C$, $x_C$ is the uplink message of the CU, $\hat{x}_S$ is the decoded version of message $x_S$, and $t_0$ denotes the processing delay. Therefore, the received signals at the BS and the D2D receiver can be expressed as
\begin{equation}
y_B(t) = g_{CB} \bar{x}_C(t) + n_B(t)
\end{equation}
and
\begin{equation}
y_D(t) = g_{CD} \bar{x}_C(t) + n_D(t),
\end{equation}
respectively.
\end{itemize}

It should be emphasized that the FD nature of the CU makes \textit{Phase-I} and \textit{Phase-II} parallel in time at the cost of self-interference. If the CU operates in the HD mode, two orthogonal channels are required to separate \textit{Phase-I} and \textit{Phase-II}, which reduces the spectral efficiency.
\subsection{SINR model}
Conditioning on $h_{SB}$, the signal-to-interference-plus-noise ratio (SINR) at the CU to decode $x_S$ is
\begin{equation}
\gamma_{SC} = \frac{p_S h_{SC}}{\beta p_C^\lambda+\sigma_C^2} = 
	\begin{cases}
		\frac{\theta h_{SC}}{h_{SB} (\beta p_C^\lambda+\sigma_C^2)}, & h_{SB} \geq \frac{\theta}{P_S}, \\
		\frac{P_S h_{SC}}{\beta p_C^\lambda+\sigma_C^2},           & h_{SB} < \frac{\theta}{P_S}.
	\end{cases}
\end{equation}

The decoding order, which has significant impact on the power allocation scheme, plays a key role in the NOMA system. In this paper, we assign the DR to decode $\hat{x}_S$ with SIC. After receiving $y_D$, the DR first regards $\hat{x}_S$ as the noise and tries to decodes $x_C$. If the decoding is successful, the DR subtracts $x_C$ from $y_D$ and then decodes $\hat{x}_S$. The SINR at the DR to decode $x_C$ is
\begin{equation}
\gamma_{CD,C} = \frac{\alpha p_C h_{CD}}{(1-\alpha) p_C h_{CD} + \sigma_D^2},
\end{equation}
and the signal-to-noise ratio (SNR) to decode $\hat{x}_S$ after SIC is
\begin{equation}
\gamma_{CD,S} = \frac{(1-\alpha) p_C h_{CD}}{\sigma_D^2}.
\end{equation}

On the other hand, the BS treats $x_S$ and $\hat{x}_S$ as interference and directly decodes $x_C$, the SINR at the BS is
\begin{equation}
\begin{split}
\gamma_{CB} & = \frac{\alpha p_C h_{CB}}{p_S h_{SB} + (1-\alpha) p_C h_{CB} + \sigma_B^2} \\
            & =  
	\begin{cases}
		\frac{\alpha p_C h_{CB}}{\theta + (1-\alpha) p_C h_{CB} + \sigma_B^2}, & h_{SB} \geq \frac{\theta}{P_S}, \\
		\frac{\alpha p_C h_{CB}}{P_S h_{SB} + (1-\alpha) p_C h_{CB} + \sigma_B^2}, & h_{SB} < \frac{\theta}{P_S}.
	\end{cases}
\end{split}
\end{equation}
\section{Performance Analysis}\label{section_iii}
In this section, we provide the performance analyses from two perspectives. On one hand, The Pareto boundary of the achievable rate region is calculated along with the corresponding power allocation strategy. On the other hand, we characterize the network performance by the joint outage probability of the cellular uplink and the cooperative D2D links, of which the exact and asymptotic expressions are presented.

\subsection{Achievable Rate Region}
For a given power allocation scheme $(\alpha, p_C)$, we have the achievable rate of the cellular uplink channel from the CU to the BS as
\begin{equation}\label{R_B}
R_B(\alpha, p_C) = \log_2 (1+\gamma_{CB}),
\end{equation}
and the achievable rate of the dual-hop D2D relay channel as
\begin{equation}\label{R_D}
R_D(\alpha, p_C) = \min \left( R_{SC}(p_C), R_{CD,S}(\alpha, p_C) \right),
\end{equation}
where 
\begin{align}
R_{SC}(p_C)&=\log_2(1+\gamma_{SC}), \\
R_{CD,S}(\alpha, p_C)&=\log_2(1+\gamma_{CD,S}).
\end{align}

From (\ref{R_B}) and (\ref{R_D}), we can see a tradeoff between the achievable rates between the cellular uplink and the cooperative D2D channels. Let us first consider two extreme cases: (i) if $\alpha = 1$, the CU would assign full power (i.e., $p_C = P_C$) for uplink transmission and the D2D link suffers an outage; (ii) if $\alpha = 0$, the CU uses full power to relay the data from the DT to the DR, and the achievable rate of the cellular uplink is zero. Therefore, we have the power allocation schemes $(0, P_C)$ and $(1, P_C)$ at the extreme points of the Pareto boundary of the achievable rate region. The values of $(\alpha, p_C)$ on the Pareto boundary can be obtained by solving the following optimization problem
\begin{align}\label{OP1}
\mathcal{OP}1: \quad \max_{\alpha, p_C} & \quad R_D(\alpha, p_C) \nonumber \\
              s.t. & \quad R_B(\alpha, p_C) = \tilde{R}_B \\
                   & \quad 0 \le \alpha \le 1, 0 \le p_C \le P_C \nonumber
\end{align}
at an arbitrary achievable rate of the cellular uplink $\tilde{R}_B \in [0, R_B^{\max}]$, where $R_B^{\max} = \log_2 \left( 1 + \frac{P_C h_{CB}}{p_S h_{SB} + \sigma_B^2} \right)$ denotes the maximum achievable rate of the cellular uplink. To solve $\mathcal{OP}1$, we provide the following lemma, which will also be used in the rest of this paper.

\begin{lemma}\label{lemma_1}
For bounded $x \in \left( {{x_{min}},{x_{max}}} \right)$, if $f\left( x \right)$ is a bounded, continuous and monotonically increasing function, and $g\left( x \right)$ is a bounded, continuous and monotonically decreasing function,  $h\left( x \right) = \min \left( {f\left( x \right),g\left( x \right)} \right)$ will be quasi-concave.
\end{lemma}
\begin{proof}
See \cite{dang2017outage} for the proof of Lemma \ref{lemma_1}.
\end{proof}
With the help of Lemma \ref{lemma_1}, we can prove that $\mathcal{OP}1$ is quasi-concave \cite{quasi_concave_ref}, and the solution is offered in Theorem \ref{theorem_1}.
\begin{theorem}\label{theorem_1}
The power allocation $(\tilde \alpha, \tilde{p}_C)$ that achieves the Pareto boundary of the rate region satisfies
\begin{equation}
\tilde{\alpha} = (1-2^{-\tilde{R}_B}) \left( 1 + \frac{p_S h_{SB} + \sigma_B^2}{\tilde{p}_C h_{CB}} \right)
\end{equation}
and
\begin{equation}
\tilde{p}_C = 
				\begin{cases}
					\frac{(p_S h_{SB} + \sigma_B^2)(2^{\tilde{R}_B} - 1)}{ h_{CB}}, & F_1 \left(\frac{(p_S h_{SB} + \sigma_B^2)(2^{\tilde{R}_B} - 1)}{ h_{CB}} \right) \geq 0, \\
					P_C, & F_1(P_C) \leq 0, \\
					\hat{p}_C, & otherwise,
				\end{cases}
\end{equation}
where $F_1(x) = 2^{-\tilde{R}_B} h_{CD} h_{CB} \beta x^{1+\lambda} + 2^{-\tilde{R}_B} h_{CD} h_{CB} \sigma_C^2 x - (1 - 2^{-\tilde{R}_B}) (p_S h_{SB} + \sigma_B^2) h_{CD} \beta x^{\lambda} - (1 - 2^{-\tilde{R}_B}) (p_S h_{SB} + \sigma_B^2) h_{CD}\sigma_C^2 - p_S h_{SC} \sigma_D^2$, and $\hat{p}_C$ satisfies $F_1(\hat{p}_C)=0$.
\end{theorem}
\begin{proof}
See Appendix \ref{proof_of_theorem_1}.
\end{proof}

\subsection{Joint Outage Probability}
An outage event occurs when neither the BS nor the DR can decode its desired message correctly. From an information-theoretic viewpoint, when the channel capacity cannot support a target rate, the failure of decoding at the receiver is doomed. Hence, the joint outage probability can be expressed by
\begin{equation}\label{jop}
\textrm{P}_{out} = \mathbb{P} \{ R_B < \eta_B, R_D < \eta_D \},
\end{equation}
where $\eta_i$ is the target rate predefined at node $i \in \{\textrm{B,D}\}$ according to a certain QoS requirement. Since the achievable rate is a monotonically increasing function of SINR, (\ref{jop}) can be rewritten as
\begin{align}
\textrm{P}_{out} & =  \mathbb{P} \{ \gamma_{CB} < \xi_B, \gamma_{CD,C} < \xi_B \bigcup \min \left( \gamma_{SC}, \gamma_{CD,S} \right) < \xi_D \} \nonumber \\
& =  1 - \mathbb{P} \{  \gamma_{CB} \geq \xi_B, \gamma_{SC} \geq \xi_D, \nonumber \gamma_{CD,C}\geq\xi_B, \gamma_{CD,S} \geq \xi_D \}
\end{align}
where $\xi_i = 2^{\eta_i}-1$, $i \in \{B,D\}$.

The exact joint outage probability is given in Theorem \ref{theorem_2}.
\begin{theorem}\label{theorem_2}
The joint outage probability of the cellular uplink and the cooperative D2D channel is
\begin{equation}
\rm{P}_{out} = 
	\begin{cases}
		1, & \alpha \le \frac{\xi_B}{1+\xi_B}, \\
		1-(\rm{P}_1 + \rm{P}_2)\rm{P}_3, & \frac{\xi_B}{1+\xi_B} < \alpha \leq 1,
	\end{cases}
\end{equation}
where
\begin{align}\label{P1}
\rm{P}_1 & = \frac{\varphi_{SC} \theta \exp\bigg[ -\frac{\theta}{P_S} \left( \frac{\xi_D (\beta p_C^\lambda+\sigma_C^2)}{\varphi_{SC} \theta} + \frac{1}{\varphi_{SB}}\right) \bigg]}{\varphi_{SB} \xi_D (\beta p_C^\lambda + \sigma_C^2) +\varphi_{SC}\theta} \exp \bigg[ -\frac{\xi_B (\theta + \sigma_B^2)}{\varphi_{CB} p_C (\alpha-\xi_B+\alpha \xi_B)} \bigg],
\end{align}
\begin{align}\label{P2}
\rm{P}_2 & = \frac{\varphi_{CB} p_C (\alpha-\xi_B+\alpha \xi_B) }{\varphi_{SB} \xi_B P_S + \varphi_{CB} p_C (\alpha-\xi_B+\alpha \xi_B)} \exp \bigg[-\frac{\xi_D (\beta p_C^\lambda + \sigma_C^2)}{\varphi_{SC} P_S}-\frac{\xi_B \sigma_B^2}{\varphi_{CB} p_C (\alpha - \xi_B +\alpha \xi_B)} \bigg] \nonumber \\
& \quad \quad \times \Bigg( 1 - \exp \bigg[ -\frac{\theta}{P_S} \big( \frac{\xi_B P_S}{\varphi_{CB} p_C (\alpha-\xi_B + \alpha \xi_B)} + \frac{1}{\varphi_{SB}} \big) \bigg]  \Bigg),
\end{align}
and
\begin{align}\label{P3}
\rm{P}_3  = 
         	\begin{cases}
         		\exp \big[ -\frac{\xi_B \sigma_D^2}{\varphi_{CD} p_C (\alpha-\xi_B+\alpha \xi_B)} \big], & \frac{\xi_B}{1+\xi_B} < \alpha \leq \frac{\xi_B \xi_D + \xi_B}{\xi_B \xi_D + \xi_B + \xi_D}, \\
         		\exp \big[ -\frac{\xi_D \sigma_D^2}{\varphi_{CD} p_C (1 - \alpha)} \big], &  \frac{\xi_B \xi_D + \xi_B}{\xi_B \xi_D + \xi_B + \xi_D} < \alpha \leq 1.
         	\end{cases}
\end{align}
\end{theorem}
\begin{proof}
See Appendix \ref{proof_of_theorem_2}.
\end{proof}

\newtheorem{corollary}{Corollary}
\begin{corollary}\label{corollary_of_theorem_2}
When $p_C$ approaches infinity, the asymptotic joint outage probability is
\begin{equation}\label{asymp_Pout}
\lim_{p_C \rightarrow +\infty} \rm{P}_{out} = 1 - \bigg[ 1-\exp \big( -\frac{\theta}{\varphi_{SB} P_S} \big) \bigg] \exp \big( -\frac{\xi_D \beta p_C^\lambda}{\varphi_{SC} P_S} \big).
\end{equation}
\end{corollary}
\begin{proof}
The proof of Corollary \ref{corollary_of_theorem_2} is straightforward, and thus is omitted.
\end{proof}

Corollary \ref{corollary_of_theorem_2} reveals that: (i) the proposed cooperative D2D scheme achieves zero-diversity. On one hand, when $\lambda=0$, $\lim_{p_C \rightarrow +\infty} \textrm{P}_{out} = 1 - \big[ 1-\exp \big( -\frac{\theta}{\varphi_{SB} P_S} \big) \big]$; on the other hand, when $\lambda \neq 0$, $\lim_{p_C \rightarrow +\infty} \textrm{P}_{out} = 1$, which implies that the considered network suffers an outage floor. (ii) the asymptotic outage performance is limited by the first-hop D2D transmission from the DT to the CU, or essentially, the RSI; (iii) the asymptotic joint outage probability is independent of the power splitting factor $\alpha$, since $\alpha$ does not impact the communication between the DT and the CU.

\section{Optimal Power Allocation for Maximizing the Minimum Achievable Rate}\label{section_iv}
In this section, the power allocation scheme $(\alpha, p_C)$ is investigated to optimize the achievable rates, taking consideration of fairness between the cellular uplink and cooperative D2D channels. The max-min criteria \cite{radunovic2007unified} is adopted to characterize the fairness, i.e., we try to find the jointly optimal $(\alpha, p_C)$ which maximizes the achievable rate of the bottleneck link. With the knowledge of global CSI, we first formulate the max-min problem. Then we discuss the conditions under which $(\alpha, p_C)$ is optimal. In the end, we prove that the problem of the max-min achievable rate is quasi-concave and provide the optimal solution.
\subsection{Problem Formulation}
When global CSI is available at the CU, the maximization problem of the minimum achievable rate for a given power allocation $(\alpha, p_C)$ can be formulated as
\begin{align}
\mathcal{OP}2: \quad  \max_{\alpha, p_C}  R_{\min}(\alpha, p_C) & = \min \Big( R_B(\alpha, p_C), R_D(\alpha, p_C) \Big) \nonumber \\
                     s.t.  \quad & 0 \leq \alpha \leq 1,  \\
                    \quad  \quad\quad 0 & \leq p_C \leq P_C,  \nonumber
\end{align}
where $R_B(\alpha, p_C)$ and $R_D(\alpha, p_C)$ are given in (\ref{R_B}) and (\ref{R_D}), respectively.
\subsection{Problem Analysis}
A key step to solve $\mathcal{OP}2$ is investigating the relation of ${R_{SC}}\left( {\alpha ,{p_C}} \right)$ and ${R_{CD,S}}\left( {\alpha ,{p_C}} \right)$ at the optimal power allocation $\left( {{\alpha ^*},p_C^*} \right)$, which is provided in the following lemma.
\begin{lemma}\label{lemma_2}
The jointly optimal $\left( {{\alpha ^*},p_C^*} \right)$ satisfies
\begin{equation}\label{constraint_of_case_1}
{R_{SC}}\left( {{\alpha ^*},p_C^*} \right) \le {R_{CD,S}}\left( {{\alpha ^*},p_C^*} \right)
\end{equation}
for ${\bar p_C} < {P_C}$, and
\begin{equation}\label{constraint_of_case_2}
{R_{SC}}\left( {{\alpha ^*},p_C^*} \right) \ge {R_{CD,S}}\left( {{\alpha ^*},p_C^*} \right)
\end{equation}
for ${\bar p_C} \ge {P_C}$, where ${\bar p_C}$ satisfies $F_2({\bar p_C}) = 0$ and $F_2(x) = h_{CD} \beta x^{1 + \lambda} +h_{CD} \sigma_C^2 x - \sigma_D^2 p_S h_{SC}$.
%\begin{equation*}
%h_{CD} \beta p_C^{1 + \lambda} +h_{CD} \sigma_C^2 p_C - \sigma_D^2 p_S h_{SC} = 0.
%\end{equation*}
\end{lemma}
\begin{proof}
See Appendix \ref{appendix_a} for the proof of Lemma 1.
\end{proof}

\subsection{Optimal Power Allocation}
According to the relation between ${\bar p_C}$ and ${P_C}$ stated in Lemma \ref{lemma_2}, the discussion of $\mathcal{OP}2$ can be divided into the following cases.

\subsubsection{Case 1} ${\bar p_C} < {P_C}$.

Based on Lemma 2, $\mathcal{OP}2$ can be equivalently reformulated as
\begin{align}
{\mathcal{OP}}2a:\quad & \mathop {\max }\limits_{\alpha ,{p_C}} {R_{\min }}\left( {\alpha ,{p_C}} \right)  = \min \left( {{R_B}\left( {\alpha ,{p_C}} \right),{R_{SC}}\left( {{p_C}} \right)} \right) \nonumber \\
\quad \quad \;\; & \quad \quad\quad \quad s.t.\quad 0 \le \alpha  \le \bar \alpha, \\
                 & \quad \quad\quad \quad \quad \quad \; \bar{p}_C \le p_C \le P_C, \nonumber
\end{align}
where $\bar \alpha  = 1 - \frac{{\sigma _D^2{p_S}{h_{SC}}}}{(\beta p_C^\lambda+\sigma_C^2) p_C h_{CD}} < 1$, and $\bar{p}_C$ is stated in Lemma \ref{lemma_2}. The first constraint in ${\mathcal{OP}}2a$ is directly derived from (\ref{constraint_of_case_1}), and the second constraint is obtained from $\bar \alpha \ge 0$. Note that ${R_{SC}}\left( {{p_C}} \right)$ is independent of $\alpha$, and then the objective function of ${\mathcal{OP}}2a$ is a non-decreasing function of $\alpha$. Without loss of generality, the optimal power splitting factor at the CU can be chosen as ${\alpha ^*} = \bar \alpha$. Substituting ${\alpha ^*} = \bar \alpha $ into ${R_{\min}}\left( {\alpha ,{p_C}} \right)$, ${\mathcal{OP}}2a$ can be simplified as, 
\begin{align}
{\mathcal{OP}}2b:\quad \mathop {\max }\limits_{{p_C}} {R_{\min }}\left( {\bar \alpha ,{p_C}} \right) & = \min \left( {{R_B}\left( {\bar \alpha ,{p_C}} \right),{R_{SC}}\left( {{p_C}} \right)} \right)\nonumber \\
\quad \quad \;\;\,\;\,s.t. & \quad {{\bar p}_C} \le {p_C} \le {P_C},
\end{align}
where ${R_B}\left( {\bar \alpha ,{p_C}} \right)$ is given in (\ref{R_B_case_1}) on the top of the next page. Obviously, ${R_B}\left( {\bar \alpha ,{p_C}} \right)$ is an increasing function of ${p_C}$ and ${R_{SC}}\left( {{p_C}} \right)$ is a decreasing function of ${p_C}$. Therefore, we can prove that ${\mathcal{OP}}2b$ is quasi-concave with the help of Lemma \ref{lemma_1}.

\begin{figure*}[ht]
\normalsize
\setcounter{equation}{30}
\begin{equation}\label{R_B_case_1}
R_B (\bar{\alpha}, p_C) = \log_2 \left( 1 + \frac{h_{CD} h_{CB} p_C (\beta p_C^\lambda + \sigma_C^2) - \sigma_D^2 h_{SC} h_{CB} p_S}{h_{SB} h_{CD} p_S (\beta p_C^\lambda + \sigma_C^2) + \sigma_D^2 h_{SC} h_{CB} p_S} \right)
\end{equation}
\setcounter{equation}{31}
\hrulefill
\vspace*{4pt}
\end{figure*}

Denote $F_3(x) = {R_B}\left( {\bar \alpha ,{x}} \right) - {R_{SC}}\left( {{x}} \right)$. There must exist a $\breve{p}_C$ which satisfies $F_3(\breve{p}_C) = 0$, and then we have the following discussion on different $\breve{p}_C$.
\begin{itemize}
\item If $\breve{p}_C < \bar{p}_C$, the minimum achievable rate is limited by $R_{SC}(\bar{\alpha}, p_C)$. The CU should use minimum transmit power to keep the RSI at a low level. Therefore, the optimal transmit power of the CU is $p_C^* = \bar{p}_C$.

\item If $\bar{p}_C \le \breve{p}_C \le P_C$, the minimum achievable rate can be maximized as $R_{\min}^{\max} \left( \bar{\alpha}, \breve{p}_C \right) = R_{B} \left( \bar{\alpha}, \breve{p}_C \right) = R_{SC} \left( \breve{p}_C \right)$, the optimal transmit power at the CU is $p_C^* = \breve{p}_C$.

\item If $\breve{p}_C > P_C$, the minimum achievable rate is limited by $R_{B}(\bar{\alpha}, p_C)$. The CU will transmit with the maximum power to improve the achievable rate of the cellular uplink, i.e., $p_C^* = P_C$.
\end{itemize}
After $p_C^*$ is obtained, we have the optimal power splitting factor as
\begin{equation}
\alpha^* = 1 - \frac{{\sigma _D^2{p_S}{h_{SC}}}}{\left[\beta (p_C^*)^\lambda+\sigma_C^2 \right] p_C^* h_{CD}}.
\end{equation}

\subsubsection{Case 2} ${\bar p_C} \ge {P_C}$.

In this case, $\mathcal{OP}2$ can be reformulated as
\begin{align}
{\mathcal{OP}}2c:\quad \mathop {\max }\limits_{\alpha ,{p_C}} {R_{\min }}\left( {\alpha ,{p_C}} \right) & = \min \left( {{R_B}\left( {\alpha ,{p_C}} \right),{R_{CD,S}}\left( {\alpha ,{p_C}} \right)} \right) \nonumber \\
s.t.  \quad & 0 \leq \alpha \leq 1,  \\
                    \quad  \quad\quad 0 & \leq p_C \leq P_C  \nonumber.
\end{align}
The objective function in ${\mathcal{OP}}2c$ is a non-decreasing function of ${p_C}$, therefore the optimal transmit power at the CU can be selected as $p_C^* = {P_C}$. Now ${\mathcal{OP}}2c$ is simplified as 
\begin{align}
{\mathcal{OP}}2d:\;\; \mathop {\max }\limits_{\alpha ,{p_C}} {R_{\min }}\left( {\alpha ,{P_C}} \right) & = \min \left( {{R_B}\left( {\alpha ,{P_C}} \right),{R_{CD,S}}\left( {\alpha ,{P_C}} \right)} \right) \nonumber \\
\quad \quad \;\;\;\,s.t.\quad & 0  \le \alpha  \le 1.
\end{align}

Similar to the discussion of ${\mathcal{OP}}2b$, we have that ${R_B}\left( {\alpha ,{P_C}} \right)$ and ${R_{CD,S}}\left( {\alpha ,{P_C}} \right)$ are monotonically increasing and decreasing function of $\alpha $. By applying Lemma \ref{lemma_1} again, we know that ${\mathcal{OP}}2d$ is also quasi-concave with respect to $\alpha $. Denote $F_4(x) = {R_B}\left( {x ,{P_C}} \right)-{R_{CD,S}}\left( {x ,{P_C}} \right)$, and it is easy to verify that $F_4(0) < 0$ and $F_4(1) > 0$. There must exist $\breve{\alpha}$ which satisfies $F_4(\breve{\alpha})=0$, such that the minimum achievable rate can be maximized as $R_{\min}^{\max} \left( \breve{\alpha}, P_C \right) = R_{B} \left( \breve{\alpha}, P_C \right) = R_{CD,S} \left( \breve{\alpha}, P_C \right)$. Therefore the optimal power splitting factor in this case is $\alpha^* = \breve{\alpha}$.

\subsection{Summary}
\begin{figure}[!t]
\centering
\includegraphics[width=4cm]{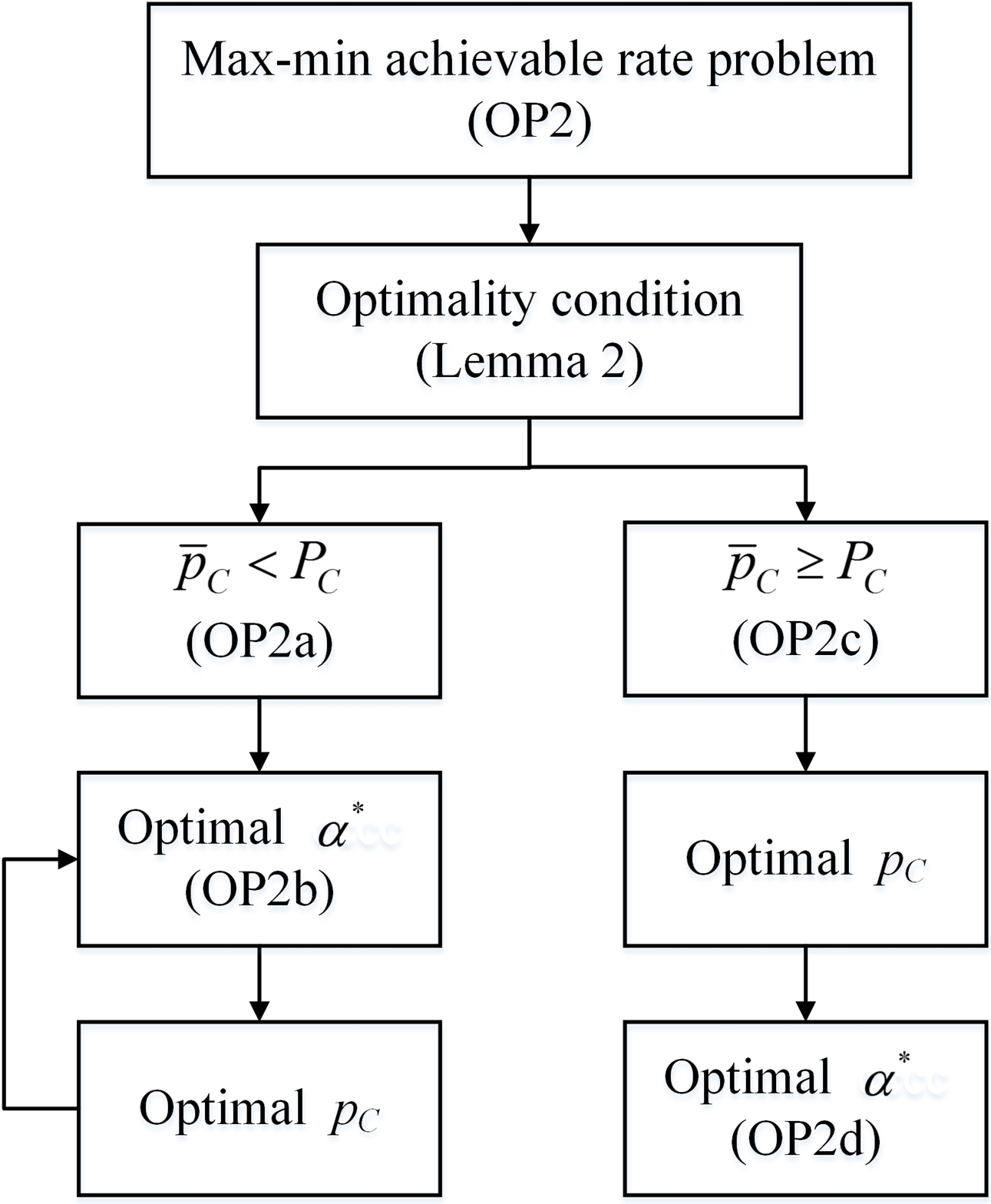}
\caption{Procedure to solve the max-min achievable rate problem.}
\label{max_min_ar}
\end{figure}

As a summary of Case 1 and Case 2, we organize the procedure to solve the max-min achievable rate problem in Fig. \ref{max_min_ar}, and propose an algorithm (denoted as Algorithm 1) to compute the jointly optimal $(\alpha, p_C)$. The idea of Algorithm 1 is sequentially identifying and maximize the bottleneck link of the considered network to meet the max-min criteria. The pseudocode of Algorithm 1 is given on the top of the next page.

\begin{algorithm}[!t]
\caption{Algorithm to compute the optimal power allocation of the max-min achievable rate problem}
\begin{algorithmic}[1]
\Require Global CSI ($h_{SB}, h_{CB}, h_{SC}, h_{CD}$); average noise power ($\sigma_B^2, \sigma_C^2, \sigma_D^2$); SIS parameters ($\beta, \lambda$); power constraint at the CU ($P_C$);
\Ensure Optimal power allocation $(\alpha, p_C)$;
\State Solve $F_2(\bar{p}_C) = 0$ to get $\bar{p}_C$;
\If{$\bar{p}_C < P_C$}
	\State Solve $F_3(x) = 0$ to get $\breve{p}_C$;
	\If{$\breve{p}_C < \bar{p}_C$}
		\State $p_C^* = \bar{p}_C$;
	\ElsIf{$\breve{p}_C > P_C$}
		\State $p_C^* = P_C$;
	\Else
		\State $p_C^* = \breve{p}_C$;
	\EndIf
	\State $\alpha^* = 1 - \frac{{\sigma _D^2{p_S}{h_{SC}}}}{\left[\beta (p_C^*)^\lambda+\sigma_C^2 \right] p_C^* h_{CD}}$;
\Else
	\State $p_C^* = P_C$;
	\State Solve $F_4(\alpha^*) = 0$ to get $\alpha^*$;
\EndIf
\end{algorithmic}
\end{algorithm}

\section{Optimal Power Allocation for Minimizing the Joint Outage Probability}\label{section_v}
When only statistical CSI is available, the CU cannot adjust transmit power or power splitting factor to maximize the achievable rate. Instead, in this paper, we consider that the CU optimizes the power allocation scheme $\left( {\alpha ,{p_C}} \right)$ to improve the outage performance. The joint outage probability minimization problem can be formulated as
\begin{align}
{\mathcal{OP}}3:\quad \mathop {\min }\limits_{\alpha ,{p_C}} & \;\quad {\rm{P}_{out}}\left( {\alpha ,{p_C}} \right)\nonumber \\
\quad \quad \;\;\,s.t. & \;\quad \frac{{{\xi _B}}}{{1 + {\xi _B}}} < \alpha  \le 1,\\
\;\quad \quad \;\;\,\quad & \quad \;0 \le {p_C} \le {P_C}, \nonumber
\end{align}
where ${\rm{P}_{out}}\left( {\alpha ,{p_C}} \right)$ is given in Theorem \ref{theorem_2}. Note that ${\rm{P}_{out}}\left( {\alpha ,{p_C}} \right)$ is a combination of exponential functions and rational fractions, which is hardly tractable. As an alternative, we will derive the upper bound of ${\rm{P}_{out}}\left( {\alpha ,{p_C}} \right)$ and loosen $\mathcal{OP}3$ to obtain a suboptimal power allocation scheme $\left( {\alpha^\circledast ,{p_C^\circledast}} \right)$. 

\subsection{Upper Bound of the Joint Outage Probability}
We use the worst-case interference approximation to obtain the upper bound of ${\rm{P}_{out}}\left( {\alpha ,{p_C}} \right)$. Considering that the interference caused by D2D transmission at the BS cannot exceed the threshold $\theta$, the SINR at the BS has a lower bound of ${\gamma _{CB}} \ge \frac{{\alpha {p_C}{h_{CB}}}}{{\theta  + \left( {1 - \alpha } \right){p_C}{h_{CB}} + \sigma_B^2}}$. Following the similar approach in Appendix \ref{proof_of_theorem_2}, the upper bound of ${\rm{P}_{out}}\left( {\alpha ,{p_C}} \right)$ for $ \frac{\xi_B}{1+\xi_B}<\alpha \le 1$ can be calculated as
\begin{equation}\label{upper_bound_of_jop}
{\rm{P}_{out}} \le 1 - \left( {{\rm{P}_1} + {{\tilde{\rm{P}}}_2}} \right){\rm{P}_3} = {\tilde{\rm{P}}_{out}},
\end{equation}
where $\tilde{\rm{P}}_2$ is given in (\ref{upper_bound_of_P_2}). $\rm{P}_1$ and $\rm{P}_3$ are given in (\ref{P1}) and (\ref{P3}), respectively. 

\begin{figure*}[ht]
\normalsize
\setcounter{equation}{36}
\begin{equation}\label{upper_bound_of_P_2}
{\tilde{\rm{P}}_2} = \exp  \left[ { - \frac{{{\xi _D}\beta (p_C^\lambda + \sigma_C^2) }}{{{\varphi _{SC}}{P_S}}} - \frac{{{\xi _B}(\theta+\sigma_B^2) }}{{{\varphi _{CB}}{p_C}\left( {\alpha  - {\xi _B} + \alpha {\xi _B}} \right)}}} \right] \times \left[ {1 - \exp \left( { - \frac{\theta }{{{\varphi _{SB}}{P_S}}}} \right)} \right],
\end{equation}
\setcounter{equation}{37}
\hrulefill
\vspace*{4pt}
\end{figure*}

The upper bound in (\ref{upper_bound_of_jop}) provides a more tractable expression. In addition, as will be shown in Section \ref{section_vi}, the derived upper bound offers a good approximation when the interference threshold $\theta$ is much smaller than the maximum transmit power at the DT, i.e., $\frac{\theta}{P_S} \rightarrow 0$. Therefore, ${\mathcal{OP}}3$ can be relaxed as
\begin{align}
{\mathcal{OP}}3a:\quad \mathop {\min }\limits_{\alpha ,{p_C}} & \;\quad {{\tilde{\rm{P}}}_{out}}\left( {\alpha ,{p_C}} \right) \nonumber \\
\quad \quad \;\;\,s.t. & \;\quad \frac{{{\xi _B}}}{{1 + {\xi _B}}} < \alpha  \le 1, \\
\;\quad \quad \;\;\,\quad & \quad \;0 \le {p_C} \le {P_C}. \nonumber
\end{align}
In general, the objective function of ${\mathcal{OP}}3a$ is not jointly concave of $(\alpha, p_C)$. However, as shown in the following, ${\mathcal{OP}}3a$ is quasi-concave.

\subsection{Optimization of the Power Splitting Factor}\label{section_v_b}
We first analyze the optimal power splitting factor ${\alpha ^ \circledast}$ for a fixed $p_C$. In order to predigest the analysis, we introduce the following variables and functions for notation convenience,
\begin{align*}
{\rm{P}_1} = A \times f\left( \alpha  \right),\quad {\tilde{\rm{P}}_2} = B \times f\left( \alpha  \right),\quad {\rm{P}_3} = g\left( \alpha  \right)
\end{align*}
where
\begin{align*}
& A = \frac{{{\varphi _{SC}}\theta \exp \left[ { - \frac{\theta }{{{P_S}}}\left( {\frac{{{\xi _D}(\beta p_C^\lambda + \sigma_C^2)}}{{{\varphi _{SC}}\theta }} + \frac{1}{{{\varphi _{SB}}}}} \right)} \right]}}{{{\varphi _{SB}}{\xi _D}(\beta p_C^\lambda +\sigma_C^2) + {\varphi _{SC}}\theta }},\\
& B = \exp \left( { - \frac{{{\xi _D}(\beta p_C^\lambda + \sigma_C^2)}}{{{\varphi _{SC}}{P_S}}}} \right) \times \left[ {1 - \exp \left( { - \frac{\theta }{{{\varphi _{SB}}{P_S}}}} \right)} \right],\\
& f\left( \alpha  \right) = \exp \left[ { - \frac{{{\xi _B}(\theta +\sigma_B^2)}}{{{\varphi _{CB}}{p_C}\left( {\alpha  - {\xi _B} + \alpha {\xi _B}} \right)}}} \right],
\end{align*}
and $g(\alpha)$ is given in (\ref{P3}).

By (\ref{upper_bound_of_jop}), the upper bound of joint outage probability can be rewritten as
\begin{equation}
\tilde{\rm{P}}_{out}\left( \alpha  \right) = 1 - \left( {A + B} \right)f\left( \alpha  \right)g\left( \alpha  \right).
\end{equation}
Taking the derivative of $\tilde{\rm{P}}_{out}\left( \alpha  \right)$, we have
\begin{align}\label{partial_derivative_wrt_alpha}
\frac{\partial }{{\partial \alpha }}{\tilde P_{out}} & \left( \alpha  \right) \nonumber \\
& =  - \left( {A + B} \right)\left[ {g\left( \alpha  \right)\frac{\partial }{{\partial \alpha }}f\left( \alpha  \right) + f\left( \alpha  \right)\frac{\partial }{{\partial \alpha }}g\left( \alpha  \right)} \right].
\end{align}
Furthermore, we have
\begin{align}
\frac{\partial }{{\partial \alpha }}f\left( \alpha  \right) & = \underbrace {\frac{{{\varphi _{CB}}{p_C}\left( {1 + {\xi _B}} \right){\xi _B}(\theta + \sigma_B^2) }}{{{{\left[ {{\varphi _{CB}}{p_C}\left( {1 + {\xi _B}} \right)\alpha  - {\varphi _{CB}}{p_C}{\xi _B}} \right]}^2}}}}_{ \buildrel \Delta \over = h\left( \alpha  \right)}f\left( \alpha  \right) \nonumber \\
& = h\left( \alpha  \right)f\left( \alpha  \right),
\end{align}
and
\begin{equation}
\frac{\partial }{{\partial \alpha }}g\left( \alpha  \right) = l\left( \alpha  \right)g\left( \alpha  \right),
\end{equation}
where
\begin{equation}
l(\alpha) = 
	\begin{cases}
		\frac{\xi_B (1+\xi_B) \sigma_D^2}{\varphi_{CD} p_C (\alpha - \xi_B + \alpha \xi_B)^2}, & \frac{\xi_B}{1+\xi_B} < \alpha \leq \frac{\xi_B \xi_D + \xi_B}{\xi_B \xi_D + \xi_B + \xi_D}, \\
         		- \frac{{{\xi _D}\sigma _D^2}}{{{\varphi _{CD}}{p_C}{{\left( {1 - \alpha } \right)}^2}}}, &  \frac{\xi_B \xi_D + \xi_B}{\xi_B \xi_D + \xi_B + \xi_D} < \alpha \leq 1.
	\end{cases}
\end{equation}

Then (\ref{partial_derivative_wrt_alpha}) can be rewritten as
\begin{equation}
\frac{\partial }{{\partial \alpha }}{\tilde{\rm{P}}_{out}}\left( \alpha  \right) =  - \left( {A + B} \right)f\left( \alpha  \right)g\left( \alpha  \right)\left[ {h\left( \alpha  \right) + l\left( \alpha  \right)} \right].
\end{equation}

For $\frac{\xi_B}{1+\xi_B} < \alpha \leq \frac{\xi_B \xi_D + \xi_B}{\xi_B \xi_D + \xi_B + \xi_D}$, we have $\frac{\partial }{{\partial \alpha }}{\tilde{\rm{P}}_{out}}\left( \alpha  \right) < 0$, i.e., $\tilde{\rm{P}}_{out}(\alpha, p_C)$ is a monotonically decreasing function in $\big( \frac{\xi_B}{1+\xi_B}, \frac{\xi_B \xi_D + \xi_B}{\xi_B \xi_D + \xi_B + \xi_D} \big]$. In this case, the optimal $\alpha$ is
\begin{equation}\label{opt_alpha_1}
\alpha^\circledast = \frac{\xi_B \xi_D + \xi_B}{\xi_B \xi_D + \xi_B + \xi_D}.
\end{equation}

For $\frac{\xi_B \xi_D + \xi_B}{\xi_B \xi_D + \xi_B + \xi_D} < \alpha \leq 1$, let $\frac{\partial }{{\partial \alpha }}{\tilde{\rm{P}}_{out}}\left( \alpha  \right) = 0$, which is equivalent to $h\left( \alpha  \right) + l\left( \alpha  \right) = 0$. It is easy to verify that equation $h\left( \alpha  \right) + l\left( \alpha  \right) = 0$ has a sole positive root which is the optimal $\alpha$,
\begin{equation}\label{opt_alpha_2}
\alpha ^ \circledast = \frac{{K + M}}{{1 + M}}
\end{equation}
where $K = \frac{{{\xi _B}}}{{1 + {\xi _B}}}$, $M = \sqrt {\frac{{{\varphi _{CD}}(\theta+\sigma_B^2) K}}{{{\varphi _{CB}}{\xi _D}\sigma _D^2}}} $. Since $0 < K < 1$, we have $K < \frac{{K + M}}{{1 + M}} < 1$, i.e., $\alpha = \frac{{K + M}}{{1 + M}}$ is a feasible solution to $\mathcal{OP}3a$. 

\subsection{Optimization of Transmit Power at the CU}\label{section_v_c}
Substituting $\alpha^\circledast$ into (\ref{upper_bound_of_jop}), ${\tilde{\rm{P}}_{out}} (\alpha^\circledast, p_C)$ can be treated as a function only depending on $p_C$,
\begin{equation}
{\tilde{\rm{P}}_{out}(\alpha^\circledast, p_C)} = 1 - j\left( {{p_C}} \right)k\left( {{p_C}} \right)
\end{equation}
where
\begin{align*}
& j\left( {{p_C}} \right) = \frac{E}{{Cp_C^\lambda  + D}} + F,\; k\left( {{p_C}} \right) = \exp \left( { - \frac{{Gp_C^{1 + \lambda } + H}}{{{p_C}}}} + I \right), \\
 & C = {\varphi _{SB}}{\xi _D}\beta ,\;D = \varphi_{SB}\xi_D\sigma_C^2 + {\varphi _{SC}}\theta,\;\\
 & E = {\varphi _{SC}}\theta \exp \left( {  \frac{{-\theta }}{{{\varphi _{SB}}{P_S}}}} \right), F = 1 - \exp \left( {  \frac{-\theta }{{{\varphi _{SB}}{P_S}}}} \right), G = \frac{{{\xi _D}\beta }}{{{\varphi _{SC}}{P_S}}},\\
 & H = \frac{\xi_B (\theta+\sigma_B^2)}{\varphi_{CB} (\alpha^\circledast - \xi_B + \alpha^\circledast \xi_B)}+ \frac{\xi_D \sigma_D^2}{\varphi_{CD} (1-\alpha^\circledast)}, I = \frac{\xi_D\sigma_C^2}{\varphi_{SC} P_S}
\end{align*}
Let ${{\partial {{\tilde P}_{out}}} \mathord{\left/{\vphantom {{\partial {{\tilde P}_{out}}} {\partial \alpha }}} \right.
 \kern-\nulldelimiterspace} {\partial \alpha }} = 0$, we have the following equation,
\begin{align}\label{eqn_p_c}
 &\lambda {C^2} FGp_C^{1 + 3\lambda } + \lambda C\left( {2DF + E} \right)Gp_C^{1 + 2\lambda } \nonumber \\
& \;\; + \lambda \left( {{D^2}FG + DEG + CE} \right)p_C^{1 + \lambda }- {C^2}FHp_C^{2\lambda } \nonumber \\
&  \;\; - C\left( {2DF + E} \right)Hp_C^\lambda  - D\left( {DF + E} \right)H = 0.
\end{align}

Denote the generalized polynomial \cite{bergelson2007distribution} on the left hand side of (\ref{eqn_p_c}) by $Q\left( {{p_C}} \right)$. It is easy to see that $Q\left( {{p_C}} \right)$ has only one sign change between the third and fourth terms. According to the \textit{Descartes' Rule of Signs} \cite{jameson2006counting}, the equation $Q\left( {{p_C}} \right)=0$ has at most one positive root. Let $p_C^\circ$ be the positive root of $Q\left( {{p_C}} \right)=0$, then we have the following discussions.
\begin{itemize}
\item Assume $p_C^ \circ  \in \emptyset$. Since $Q\left( 0 \right) =  - D\left( {DF + E} \right)H < 0$, we know ${\left. {\frac{\partial }{{\partial {p_C}}}{{\tilde P}_{out}}} \right|_{{p_C} = 0}} < 0$ for ${p_C} \ge 0$. Therefore, ${\tilde{\rm{P}}_{out}}$ is a monotonically decreasing function of ${p_C}$. The optimal transmit power at the CU is
\begin{equation}
p_C^\circledast = P_C.
\end{equation}

\item If $p_C^ \circ  \notin \emptyset$ and $p_C^ \circ  \notin \left[ {0,{P_C}} \right]$, ${\tilde{\rm{P}}_{out}}$ is a decreasing function of $p_C$ in the feasible region of $\mathcal{OP}3a$. In this case, the optimal solution is $p_C^\circledast = P_C$.

\item If $p_C^ \circ  \notin \emptyset$ and $p_C^ \circ  \in \left[ {0,{P_C}} \right]$, then $\tilde{\rm{P}}_{out}$ is a decreasing function of $p_C$ for $0 \le {p_C} \le p_C^ \circ $ and an increasing function for $p_C^ \circ  \le {p_C} \le {P_C}$. Hence, the optimal transmit power at the CU is 
\begin{equation}
p_C^\circledast = p_C^\circ.
\end{equation}
\end{itemize}

\subsection{Summary}
Following the analyses in Section \ref{section_v_b} and \ref{section_v_c}, the suboptimal solution to $\mathcal{OP}3$ can be summarized as
\begin{equation}
(\alpha^\circledast, p_C^\circledast)=\left( \max\left( \frac{\xi_B \xi_D + \xi_B}{\xi_B \xi_D + \xi_B + \xi_D},\frac{{K + M}}{{1 + M}} \right), \min\left( p_C^\circ, P_C \right) \right).
\end{equation}
The procedure to solve the minimum joint outage probability problem is illustrated in Fig. \ref{fig_mjop}. Apparently, the optimal power splitting factor is independent of the transmit power at the CU, therefore $\alpha ^ \circledast$ and ${p_C^\circledast}$ can be individually obtained.
\begin{figure}
\centering
\includegraphics[width=4cm]{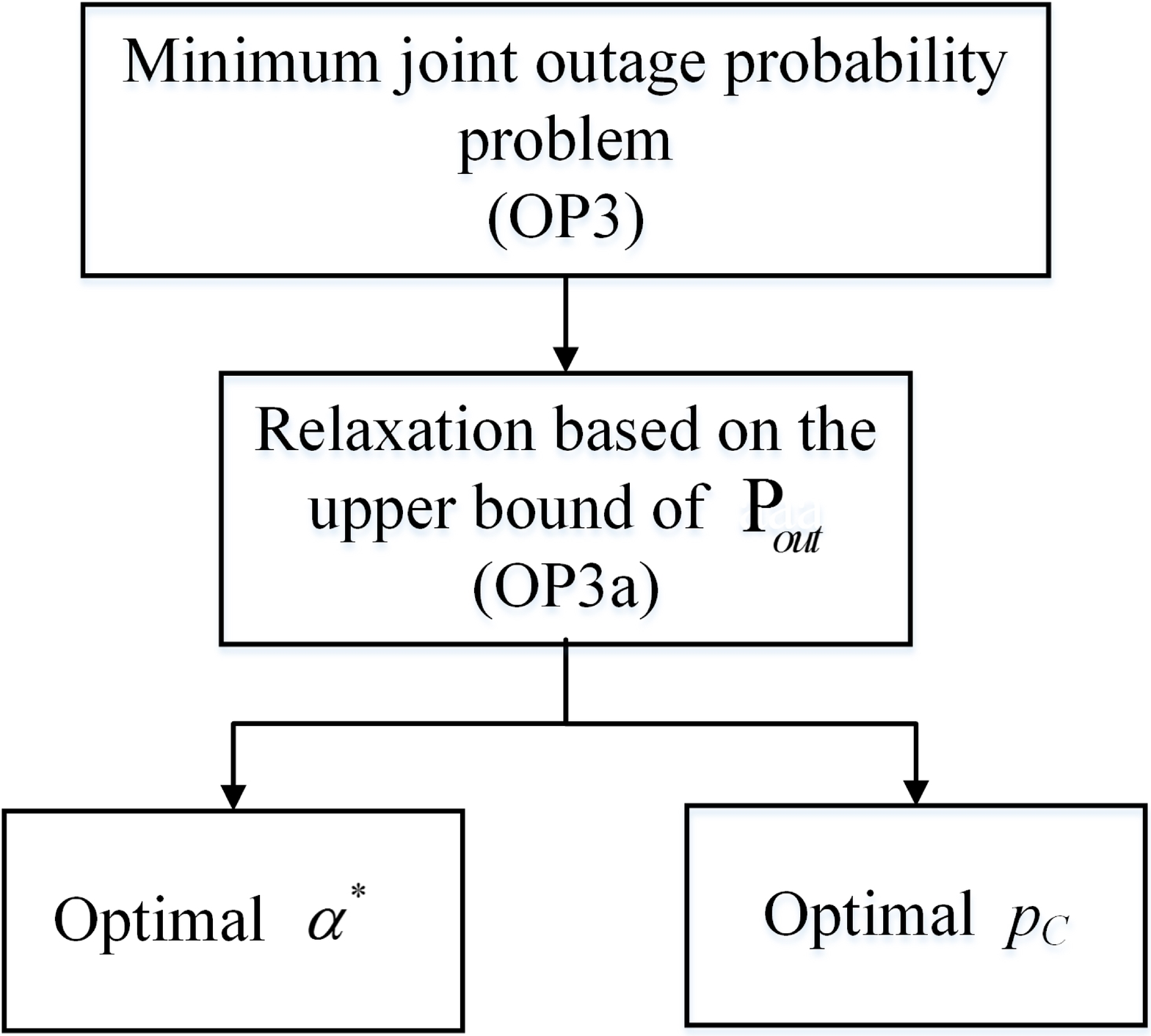}
\caption{Procedure to solve the minimum joint outage probability problem.}
\label{fig_mjop}
\end{figure}

\section{Numerical Results}\label{section_vi}

In this section, we use numerical simulations to verify the performance analysis and evaluate the proposed power allocation algorithms. The channel coefficient is independently realized in each simulation according to the Gaussian distribution. The main simulation parameters are listed in Table I, other involved parameters will be stated in each simulation. In addition, all non-linear equations are solved with the bisection method \cite{burden1985numerical}.

\begin{table}\label{table_sim_para}
\caption{Main Simulation Parameters}
\begin{center}
\begin{tabular}[!t]{l | c}
\hline
\textbf{Parameters}               & \textbf{Value} \\
\hline
Carrier center frequence & 2GHz \\
Channel bandwidth        & 180kHz \\
Peak transmit power of the users & 23dBm \\
Receiver noise density   & -174dBm \\
Cell radius              & 200m \\
Distance between the DT and the DR & 150-300m \\
Minimum distance between the users and the BS & 30m \\
Decay factor of the path-loss   & 3.8\\
\hline
\end{tabular}
\end{center}

\end{table}

\begin{figure}[!t]
\centering
\includegraphics[width=7cm]{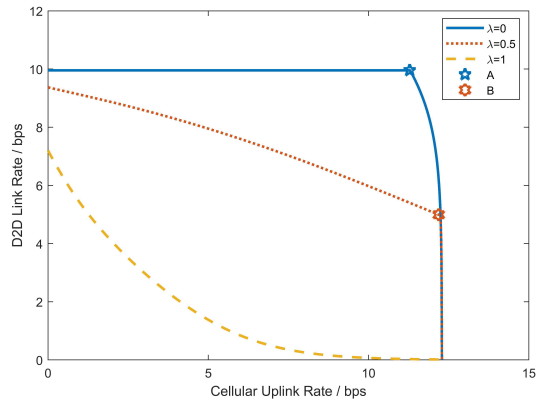}
\caption{Achievable rate region ($h_{SB}=h_{SC}=h_{CB}=h_{CD}=0.5$, $\beta = 1$, $P_S=23\rm{dBm}$, $P_C = 23\rm{dBm}$).}
\label{fig_rate_region}
\end{figure}

Fig. \ref{fig_rate_region} shows the achievable rate region of the proposed cooperative D2D network with different $\lambda$, which is the SIS parameter defined in (\ref{rsi_model}). We can see that with the decrease of $\lambda$, the proposed network has a larger achievable rate region. The maximum achievable rate of the cellular uplink for different $\lambda$ is the same, since the uplink rate is not affected by RSI. However, the maximum achievable rate of the cooperative D2D link decrease as $\lambda$ grows due to the strengthened RSI. At point \textsf{A} for $\lambda=0$ or \textsf{B} for $\lambda=0.5$, the first and second hop of the cooperative D2D link achieves the identical rate. Further increase of $R_B$ requires larger $p_C$ or $\alpha$, which will cause rapid descent of $R_D$.

Fig. \ref{fig_2} shows the joint outage probability of the considered system with a fixed power splitting factor   at the CU. On one hand, we observe that the curves of theoretical analysis perfectly match the curves of Monte Carlo simulation results, which confirms our analytical results in Section \ref{section_iii}. On the other hand, the curves of the upper bound of the joint outage probability almost overlap the curves of the exact joint outage probability, which indicates that the upper bound shown in (\ref{upper_bound_of_jop}) can be regarded as an accurate approximation of the exact joint outage probability.

\begin{figure}[!t]
\centering
\includegraphics[width=7cm]{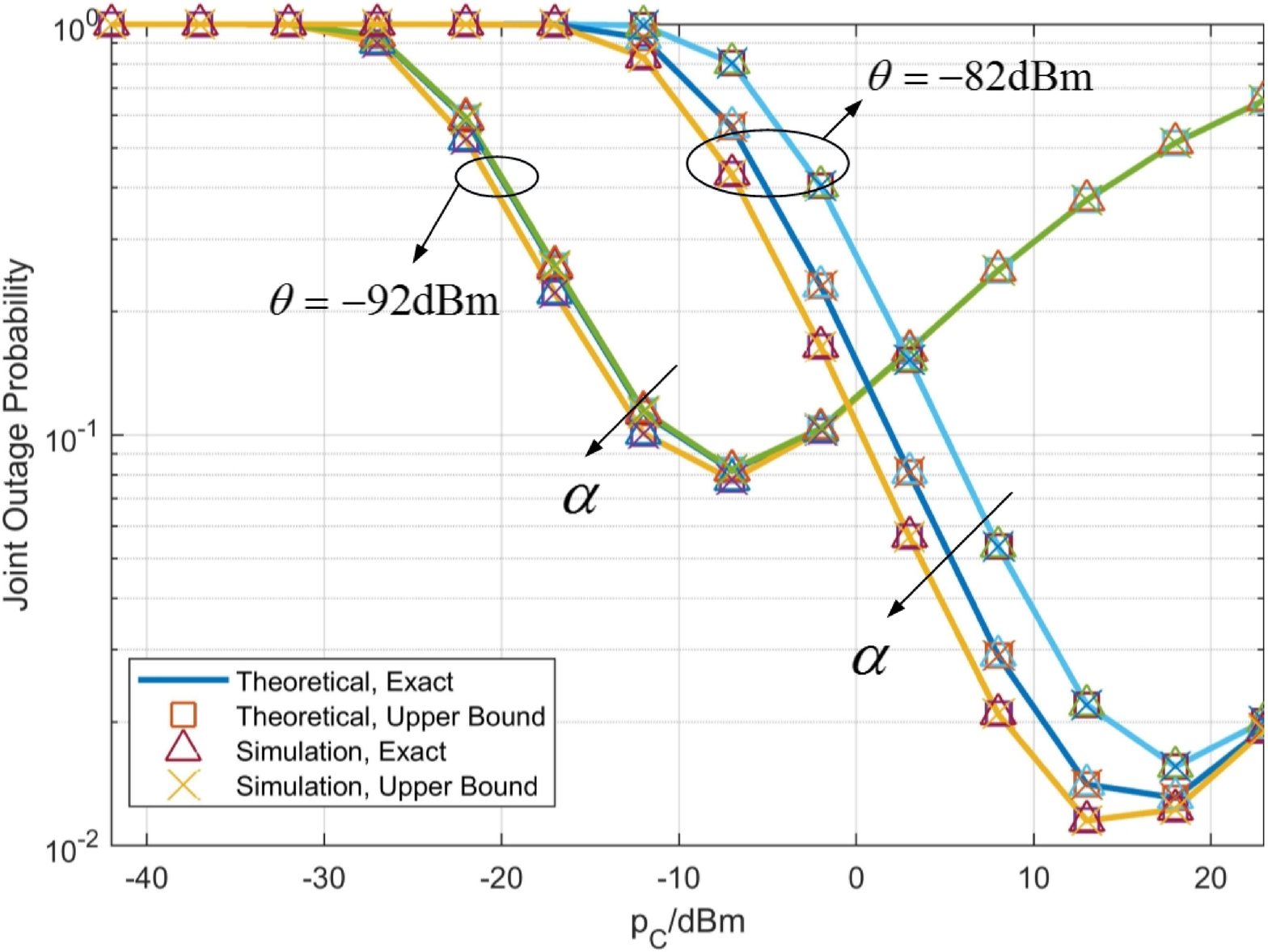}
\caption{Joint outage probability versus $p_C$ ($\alpha = [0.6,0.7,0.8]$, $\lambda=0.1$, $\beta = 1$, $P_S=23\rm{dBm}$, $\eta_B=\eta_D=1$).}
\label{fig_2}
\end{figure}

Roughly speaking, the curve of the joint outage probability is a ``V" shape. With the increase of $p_C$, the joint outage probability first decreases due to the improvement of SINR/SNR. However, a further increase of $p_C$ causes more sever RSI at the CU, and then leads to the growth of joint outage probability, since the joint outage probability is dominated by the link between the DT and the CU in the high transmit power region. Furthermore, the curves of joint outage probability with different $\alpha$ converges when $p_C$ is high enough. This can be explained by asymptotic analysis of $\rm{P}_{out}$ in Section \ref{section_iii}. For fixed and relatively smaller $p_C$, the joint outage probability decreases with $\alpha$, since the joint outage probability is limited by the cellular uplink transmission from the CU to the BS. Besides, a larger $\theta$ will loosen the transmit power constraint at the D2D transmitter and leads to a lower minimum joint outage probability at the cost of higher transmit power of the CU. 

\begin{figure}[!t]
\centering
\includegraphics[width=7cm]{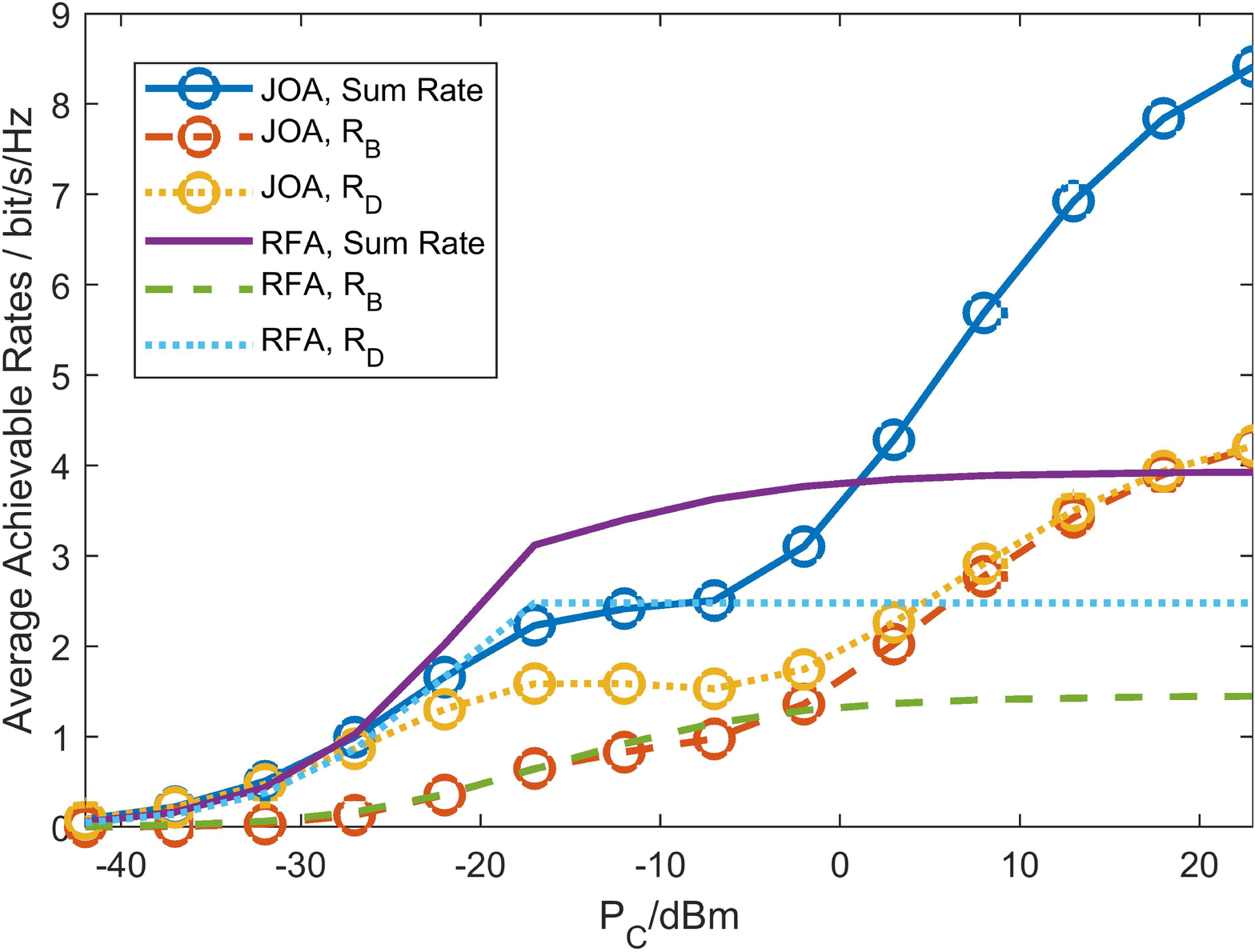}
\caption{Comparison of average achievable rate with difference power allocation algorithms ($\theta=-92\rm{dBm}$, $\lambda=0.1$, $\beta = 1$, $P_S=23\rm{dBm}$, $\eta_B=\eta_D=1$).}
\label{fig_3}
\end{figure}
Fig. \ref{fig_3} illustrates the average achievable rate with the proposed joint optimization algorithm (JOA). For comparison, we adopt the random-$\alpha$  fixed-$p_C$  algorithm (RFA) as a benchmark, where the CU uniformly selects $\alpha$ from $[0,1]$ and transmits with the maximum power. From a sum rate viewpoint, the RFA outperforms the JOA when $P_C<0\rm{dBm}$. When $P_C>0\rm{dBm}$, the RFA introduces severe self-interference due to maximum transmit power at the CU. Meanwhile, the randomly chosen $\alpha$ restricts the achievable rate at the BS. Therefore, the sum rates of RFA reaches a plateau rapidly and results in the waste of transmit power at the CU. With the JOA, the CU can dynamically maximize the achievable rate of the bottleneck link according to $p_C$. From the decoding order at the BS and the DUE receiver, we know that the considered system is limited by the cellular uplink channel capacity. As shown in Fig. \ref{fig_3}, $R_B$ with JOA grows monotonically as $P_C$ increases. For $P_C>10\rm{dBm}$, $R_B$ and $R_D$ converge to a same value, which confirms the validity of the proposed JOA. In addition, the JOA achieves a much higher sum rate than the RFA does, which demonstrates that the transmit power at the CU can be more effectively utilized with JOA.

\begin{figure}[!t]
\centering
\includegraphics[width=7cm]{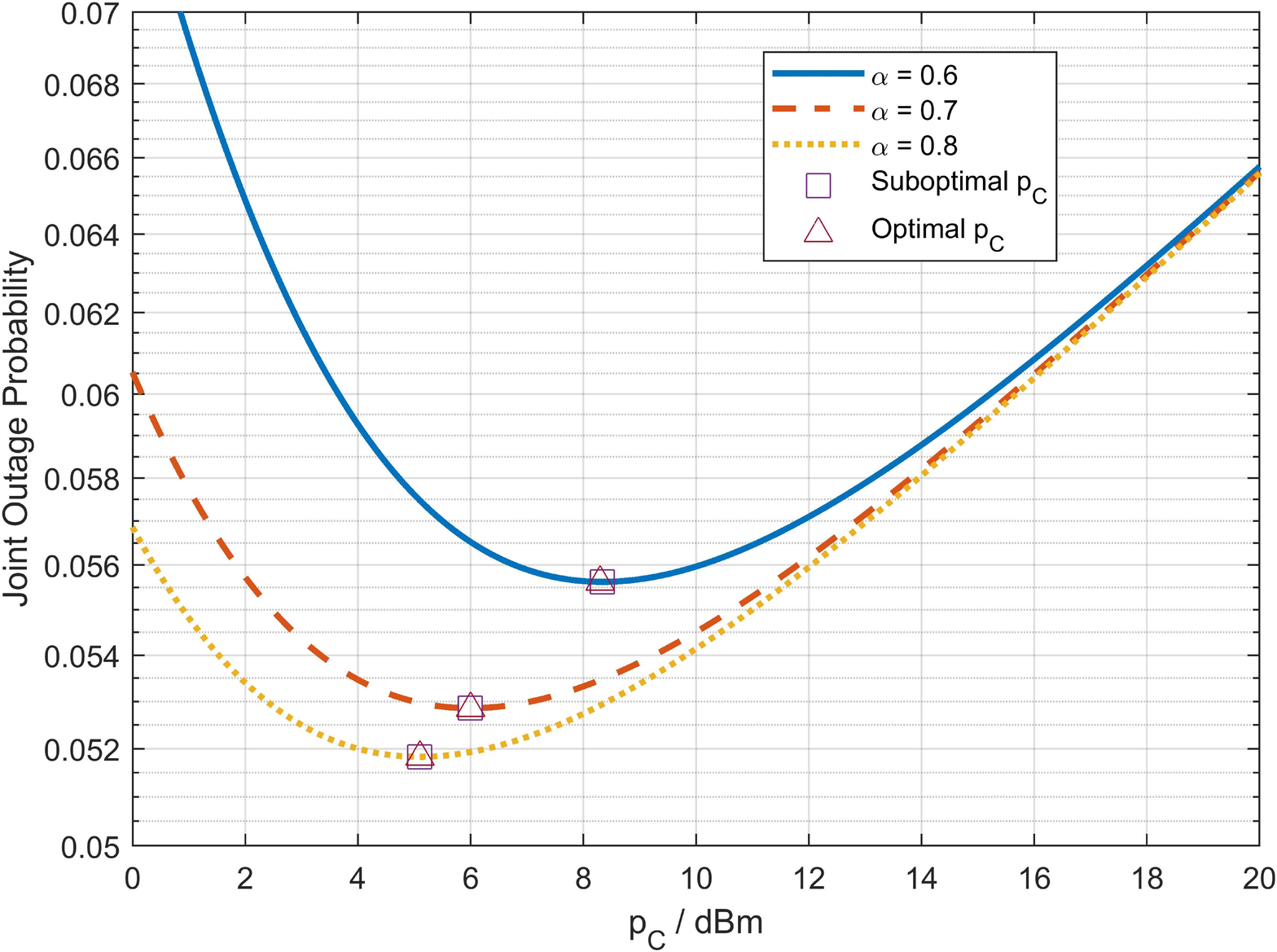}
\caption{Joint outage probability with different $\alpha$ ($\theta=-92\rm{dBm}$, $\lambda=0.1$, $\beta = 1$, $P_S=23\rm{dBm}$, $\eta_B=\eta_D=1$).}
\label{fig_4}
\end{figure}

\begin{figure}[!t]
\centering
\includegraphics[width=7cm]{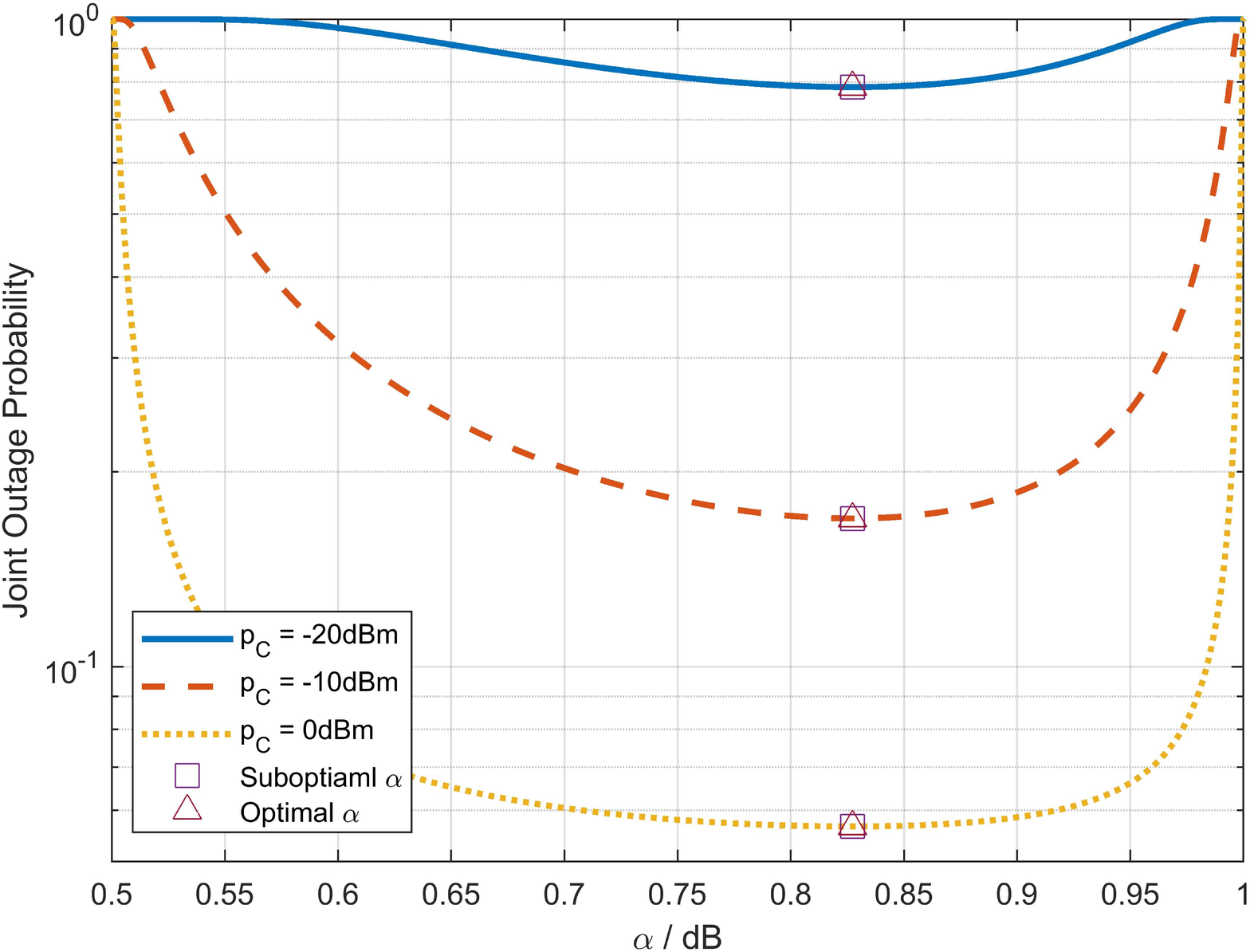}
\caption{Joint outage probability with different $p_C$ ($\theta=-92\rm{dBm}$, $\lambda=0.1$, $\beta = 1$, $P_S=23\rm{dBm}$, $\eta_B=\eta_D=1$).}
\label{fig_5}
\end{figure}

Fig. \ref{fig_4} and Fig. \ref{fig_5} provide the simulation results of the suboptimal power allocation in the sense of joint outage probability minimization. The optimal solutions are carried out by exhaustive search. It can be observed that for fixed $\alpha$ or $p_C$, the suboptimal solutions is very close to the optimal solutions. As analyzed in Section \ref{section_v}, $\alpha$ and $p_C$ can be optimized separately, and therefore the effectiveness of the proposed suboptimal power allocation is verified.

\begin{figure}[!t]
\centering
\includegraphics[width=7cm]{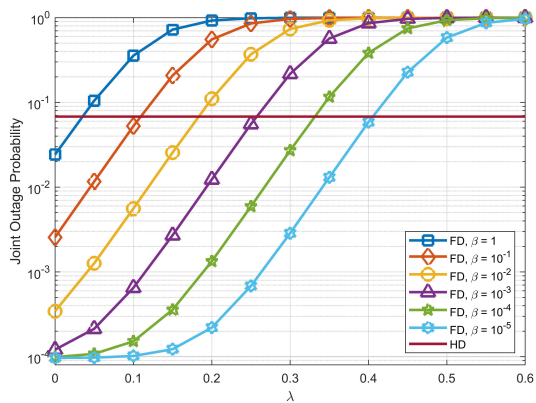}
\caption{Joint outage probability versus $\lambda$ ($\theta=-92\rm{dBm}$, $\alpha=0.95$, $p_C = 23\rm{dBm}$ $P_S=23\rm{dBm}$, $\eta_B=\eta_D=1$).}
\label{fig_6}
\end{figure}

\begin{figure}[!t]
\centering
\includegraphics[width=7cm]{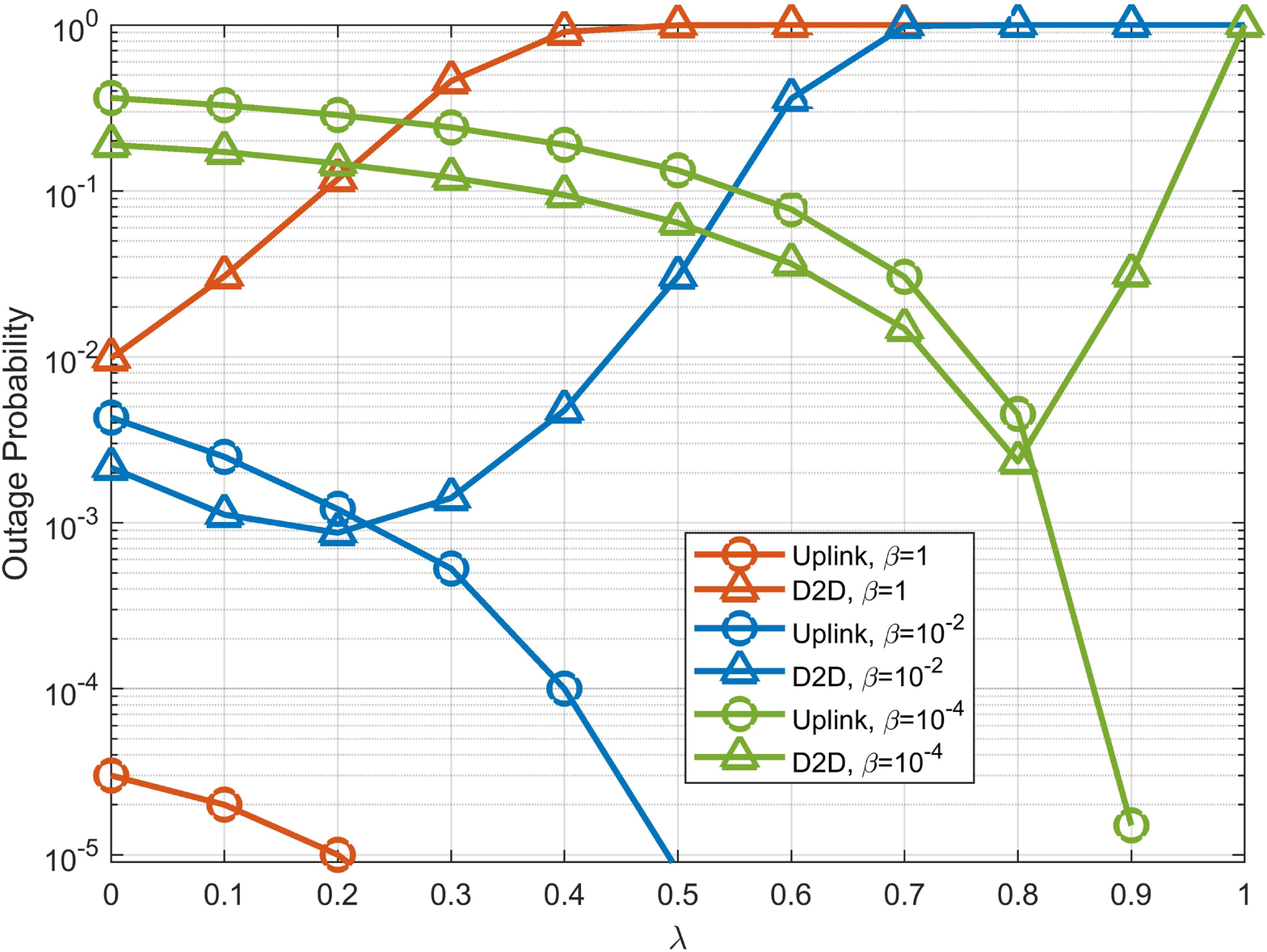}
\caption{Joint outage probability versus $\lambda$ ($\theta=-92\rm{dBm}$, $\alpha=0.95$, $p_C = 23\rm{dBm}$ $P_S=23\rm{dBm}$, $\eta_B=\eta_D=1$).}
\label{fig_7}
\end{figure}

In order to investigate the relation between the network performance and the performance of SIS, we also simulate the joint outage probability as a function of $\lambda$, with different $\beta$. The HD system is adopted as a benchmark. We can see that if the TRR is high enough, the FD network outperforms the HD counterpart. By (\ref{trr}), the TRR at each cross point where the FD and HD networks achieve the same outage performance is 130dB. In other words, the advantage of the FD mode over the HD mode lies on the TRR rather than the unilateral value of $\lambda$ and $\beta$. However, different $\lambda$ and $\beta$ provide distrinct tradeoff between the outage performance of the cellular uplink and the cooperative D2D link. In Fig. \ref{fig_7}, we present the outage probability for the cellular uplink and the cooperative D2D link with TRR fixed to 130dB. We can see that the outage probability of the cellular uplink decreases with $\lambda$, and the outage probability of the cooperative D2D link increases as $\lambda \rightarrow 1$. For $\beta=10^{-2}$ and $\beta = 10^{-4}$, two reverse points at $\lambda = 0.2$ and $\lambda=0.8$ are observed. This phenomenon can be explained by investigating the relation among the TRR, $p_C$, $\lambda$ and $\beta$. For fixed $\beta$ and TRR, we have $p_C = \sqrt[1-\lambda]{TRR\times\beta}$, which indicates that $p_C$ is an increasing function of $\lambda$. Note that $R_B(\alpha, p_C)$ is also an increasing function of $\lambda$, and therefore the outage probability of the cellular uplink monotonically decreases with $\lambda$. However, with the increase of $p_C$, the outage probability of the first-hop D2D link from the DT to the CU is worsen due to strengthened RSI, meanwhile the outage probability of the second-hop D2D link from the CU to the DR is improved due to elevated SNR. In addition, the outage probabilities of the cellular uplink and the cooperative D2D link are close when $\lambda \rightarrow 0$, but diverge when $\lambda \rightarrow 1$, which implies that smaller $\lambda$ provides better fairness between the cellular uplink and the cooperative D2D link.

\section{Conclusion}\label{section_vii}
In this paper, we proposed a cooperative underlay D2D network, where the cellular user is assigned as an FD relay with superposition coding and the D2D receiver performs successive interference cancellation to decode the desired signal. Both achievable rate region and joint outage probability were analyzed. To optimize the network performance, two power allocation schemes were proposed in the sense of max-min achievable rate and minimizing the upper bound of the joint outage probability. The correctness of theoretical analysis and the validity of power allocation schemes have been verified by numerical simulations, which reveals the superiority of the proposed FD cooperative D2D network.

\appendices
\section{Proof of Theorem \ref{theorem_1}}\label{proof_of_theorem_1}
By the first constraint in (\ref{OP1}), we can express $\alpha$ as a function of $p_C$,
\begin{equation}\label{revised_alpha}
\alpha = (1-2^{-\tilde{R}_B}) \left( 1 + \frac{p_S h_{SB} + \sigma_B^2}{p_C h_{CB}} \right).
\end{equation}
Then, $R_{CD,S}(\alpha, p_C)$ can be rewritten as (\ref{revised_R_CD_S})

\begin{equation}\label{revised_R_CD_S}
R_{CD,S}(p_C) = \log_2 \left( 1 - \frac{h_{CD} (p_S h_{SB} + \sigma_B^2) (1-2^{-\tilde{R}_B})}{h_{CB} \sigma_D^2} + \frac{h_{CD} 2^{-\tilde{R}_B}}{\sigma_D^2} p_C \right).
\end{equation}

Since $0 \leq \alpha \leq 1$, we have
\begin{equation}
p_C \geq \frac{(p_S h_{SB} + \sigma_B^2)(2^{\tilde{R}_B} - 1)}{ h_{CB}}.
\end{equation}
In addition, $\frac{(p_S h_{SB} + \sigma_B^2)(2^{\tilde{R}_B} - 1)}{ h_{CB}} \leq P_C$ is guaranteed by $\tilde{R}_B \leq R_B^{\max}$. Therefore we can reformulate $\mathcal{OP}1$ as follows,
\begin{align}\label{revised_OP1}
 \max_{\alpha, p_C} & \quad \min \left( R_{SC}(p_C), R_{CD,S}(p_C) \right) \nonumber \\
               s.t. & \frac{(p_S h_{SB} + \sigma_B^2)(2^{\tilde{R}_B} - 1)}{ h_{CB}} \leq p_C \leq P_C.
\end{align}
Obviously, $R_{CD,S}(p_C)$ is a monotonically increasing function of $p_C$, and $R_{SC}(p_C)$ is a monotonically decreasing function of $p_C$. By Lemma \ref{lemma_1}, the optimization problem in (\ref{revised_OP1}) is quasi-concave. Denote $F_1(x) = R_{CD,S}(x) - R_{SC}(x)$, then the solution to (\ref{revised_OP1}) can be divided in three cases:

\textit{Case 1}: $F_1 \left( \frac{(p_S h_{SB} + \sigma_B^2)(2^{\tilde{R}_B} - 1)}{ h_{CB}} \right) \geq 0$. In this case, the cooperative D2D link is restricted by the first-hop from the DT to the CU. The CU has to limit the transmit power to avoid severe RSI. Therefore $p_C$ is chosen to meet the lower bound as $p_C = \frac{(p_S h_{SB} + \sigma_B^2)(2^{\tilde{R}_B} - 1)}{ h_{CB}}$.

\textit{Case 2}: $F_1 \left( P_C \right) \leq 0$. In this case, the bottleneck link in the cooperative D2D channel is the second-hop from the CU to the DT. Therefore the CU uses the highest transmit power to achieve the Pareto boundary, i.e., $p_C = P_C$.

\textit{Case 3}: $F_1 \left( \frac{(p_S h_{SB} + \sigma_B^2)(2^{\tilde{R}_B} - 1)}{ h_{CB}} \right) < 0$ and $F_1 \left( P_C \right) > 0$. In this case, there must exist a $\hat{p}_C$ such that the achievable rate of the cooperative D2D channel can be maximized as $R_D^{\max} = R_{CD,S}(\hat{p}_C) = R_{SC}(\hat{p}_C)$. The uniqueness of $\hat{p}_C$ is guaranteed by the monotonicity of $R_{SC}(p_C)$ and $R_{CD,S}(p_C)$.

Substituting the $p_C$ in Cases 1-3 into (\ref{revised_alpha}), the proof of Theorem \ref{theorem_1} is completed.

\section{Proof of Theorem \ref{theorem_2}}\label{proof_of_theorem_2}
Conditioning on $\alpha$, the discussion of $\textrm{P}_{out}$ can be divided into the following two cases.

\textit{Case A}: $\alpha \leq \frac{\xi_B}{1+\xi_B}$. In this case, we have $\frac{\alpha}{1-\alpha} \leq \xi_B$. On one hand, we know that $\gamma_{CD,C} < \frac{\alpha}{1-\alpha} \leq \xi_B$, the DR will fail to decode $x_C$, and thus cannot perform SIC to further decode $x_S$. On the other hand, we have $\gamma_{CB} < \frac{\alpha}{1-\alpha} \leq \xi_B$, which suggests that the decoding at the BS also fails. Therefore, the joint outage probability $\textrm{P}_{out} = 1$.

\textit{Case B}: $\alpha > \frac{\xi_B}{1+\xi_B}$. Considering that $\gamma_{CD,C}$ and $\gamma_{CD,S}$ are independent of $\gamma_{CB}$ and $\gamma_{SC}$, $\textrm{P}_{out}$ can be rewritten as
\begin{align}
\textrm{P}_{out}  & = 1 - \mathbb{P} \{ \gamma_{CB} \geq \xi_B, \gamma_{SC} \geq \xi_D \} \nonumber \\
& \quad\quad\quad\quad\times \mathbb{P} \{ \gamma_{CD,C} \geq \xi_B, \gamma_{CD,S} \geq \xi_D \} \nonumber \\
    & =1 - \mathbb{E}_{h_{SB}} \big[ \mathbb{P} \{ \gamma_{CB} \geq \xi_B, \gamma_{SC} \geq \xi_D \lvert h_{SB}\} \big] \nonumber \\
    & \quad\quad\quad\quad \times \mathbb{P} \{ \gamma_{CD,C} \geq \xi_B, \gamma_{CD,S} \geq \xi_D \} \nonumber \\
    & = 1 - \bigg( \underbrace{\mathbb{P} \{ \gamma_{CB} \geq \xi_B, \gamma_{SC} \geq \xi_D, h_{SB} \geq \frac{\theta}{P_S} \}}_{\triangleq \textrm{P}_1} \nonumber \\
    & \quad\quad\quad\quad + \underbrace{\mathbb{P} \{ \gamma_{CB} \geq \xi_B, \gamma_{SC} \geq \xi_D, h_{SB} < \frac{\theta}{P_S}\}}_{\triangleq \textrm{P}_2} \bigg) \nonumber \\
    & \quad\quad\quad\quad \times \underbrace{\mathbb{P} \{ \gamma_{CD,C} \geq \xi_B, \gamma_{CD,S} \geq \xi_D \}}_{\triangleq \textrm{P}_3}.
\end{align}

$\rm{P}_1$ can be further expanded as
\begin{equation}\label{P_1_expanded}
\rm{P}_1 = \underbrace{\mathbb{P}\{\gamma_{CB}\geq\xi_B\}}_{\rm{Q}_1} \underbrace{\mathbb{P}\{\gamma_{SC}\geq\xi_D \big| h_{SB} \geq \frac{\theta}{P_S}\} \mathbb{P}\{h_{SB}\geq\frac{\theta}{P_S}\}}_{\rm{Q}_2}.
\end{equation}
Following the Rayleigh fading assumption, we have the probability density function (pdf) of $h_{ij}$ as $f_{h_{ij}}(x) = \frac{1}{\varphi_{ij}} e^{-\frac{x}{\varphi_{ij}}}$, and thus $\rm{Q}_1$ can be trivially obtained as
\begin{align}\label{Q_1}
\rm{Q}_1 & = \mathbb{P} \bigg\{ h_{CB} \geq \frac{\xi_B (\theta + \sigma_B^2)}{p_C (\alpha-\xi_B+\alpha \xi_B)} \bigg\} \nonumber \\
         & = \int_{\frac{\xi_B (\theta + \sigma_B^2)}{p_C (\alpha-\xi_B+\alpha \xi_B)}}^{+\infty} \frac{1}{\varphi_{CB}} e^{-\frac{h_{CB}}{\varphi_{CB}}} \rm{d} \it{h}_{CB} \\
         & = \exp \bigg[ -\frac{\xi_B (\theta + \sigma_B^2)}{\varphi_{CB} p_C (\alpha-\xi_B+\alpha \xi_B)} \bigg].\nonumber
\end{align}
And $\rm{Q}_2$ can be calculated as
\begin{align}\label{Q_2}
\rm{Q}_2 & = \mathbb{P} \bigg\{ h_{SC} \geq \frac{\xi_D h_{SB} (\beta p_C^\lambda+\sigma_C^2)}{\theta} \big| h_{SB} \geq \frac{\theta}{P_S} \bigg\} \mathbb{P}\{h_{SB}\geq\frac{\theta}{P_S}\} \nonumber \\
         & = \int_{\frac{\theta}{P_S}}^{+\infty} \frac{1}{\varphi_{SB}} \exp \left[ -\frac{\xi_D h_{SB} (\beta p_C^\lambda+\sigma_C^2)}{\varphi_{SC} \theta} \right] \exp \left( -\frac{h_{SB}}{\varphi_{SB}} \right) \rm{d} \it{h}_{SB} \nonumber \\
         & = \frac{\varphi_{SC} \theta \exp\bigg[ -\frac{\theta}{P_S} \left( \frac{\xi_D (\beta p_C^\lambda+\sigma_C^2)}{\varphi_{SC} \theta} + \frac{1}{\varphi_{SB}}\right) \bigg]}{\varphi_{SB} \xi_D (\beta p_C^\lambda + \sigma_C^2) +\varphi_{SC}\theta}. 
\end{align}
Substituting (\ref{Q_1}) and (\ref{Q_2}) back into (\ref{P_1_expanded}), we have $\rm{P}_1$ in (\ref{P1}).

Similarly, $\textrm{P}_2$ can be expanded as
\begin{equation}\label{P_2_expanded}
\rm{P}_2 = \underbrace{\mathbb{P}\{\gamma_{SC}\geq\xi_D\}}_{\rm{Q}_3} \underbrace{\mathbb{P}\{\gamma_{CB}\geq\xi_B \big| h_{SB} < \frac{\theta}{P_S}\} \mathbb{P}\{h_{SB}<\frac{\theta}{P_S}\}}_{\rm{Q}_4}.
\end{equation}

$\rm{Q}_3$ can be computed as
\begin{align}\label{Q_3}
\rm{Q}_3 & = \mathbb{P} \bigg\{ h_{SC} \geq \frac{\xi_D (\beta p_C^\lambda+\sigma_C^2)}{P_S} \bigg\} \nonumber \\
         & = \int_{\frac{\xi_D (\beta p_C^\lambda+\sigma_C^2)}{P_S}}^{+\infty} \frac{1}{\varphi_{SC}} e^{-\frac{h_{SC}}{\varphi_{SC}}} \rm{d} \it{h}_{SC} \\
         & = \exp \bigg[ -\frac{\xi_D (\beta p_C^\lambda+\sigma_C^2)}{\varphi_{SC} P_S} \bigg].\nonumber
\end{align}

$\rm{Q}_4$ can be calculated as
\begin{align}\label{Q_4}
\rm{Q}_4 & = \mathbb{P} \bigg\{ h_{CB} \geq \frac{\xi_B (P_S h_{SB} + \sigma_B^2)}{p_C (\alpha-\xi_B+\alpha \xi_B)} \big| h_{SB} < \frac{\theta}{P_S} \bigg\} \mathbb{P}\{h_{SB} < \frac{\theta}{P_S}\} \nonumber \\
         & = \int_0^{\frac{\theta}{P_S}} \frac{1}{\varphi_{SB}} \exp \left[ -\frac{\xi_B (P_S h_{SB} + \sigma_B^2)}{p_C (\alpha-\xi_B+\alpha \xi_B)} \right] \exp \left( -\frac{h_{SB}}{\varphi_{SB}} \right) \rm{d} \it{h}_{SB} \\
         & = \frac{\varphi_{CB} p_C (\alpha-\xi_B+\alpha \xi_B) }{\varphi_{SB} \xi_B P_S + \varphi_{CB} p_C (\alpha-\xi_B+\alpha \xi_B)} \nonumber \\
& \quad\quad \times \exp \bigg[-\frac{\xi_B \sigma_B^2}{\varphi_{CB} p_C (\alpha - \xi_B +\alpha \xi_B)} \bigg] \nonumber \\
& \quad \quad \times \Bigg( 1 - \exp \bigg[ -\frac{\theta}{P_S} \big( \frac{\xi_B P_S}{\varphi_{CB} p_C (\alpha-\xi_B + \alpha \xi_B)} + \frac{1}{\varphi_{SB}} \big) \bigg]  \Bigg).
\end{align}
Substituting (\ref{Q_3}) and (\ref{Q_4}) back into (\ref{P_2_expanded}), we have $\rm{P}_2$ in (\ref{P2}).

$\rm{P}_3$ can be rewritten as
\begin{equation}
{\rm{P}}_3 = \mathbb{P} \{h_{CD} \geq \frac{\xi_B \sigma_D^2}{p_C (\alpha-\xi_B+\alpha \xi_B)}, h_{CD} \geq \frac{\xi_D \sigma_D^2}{p_C (1 - \alpha)}\}.
\end{equation}
We can see that the expression of $\rm{P}_3$ is segmented by $\alpha$. If $\frac{\xi_B \sigma_D^2}{p_C (\alpha-\xi_B+\alpha \xi_B)} \geq \frac{\xi_D \sigma_D^2}{p_C (1 - \alpha)}$, which is equivalent to $\frac{\xi_B}{1+\xi_B} < \alpha \leq \frac{\xi_B \xi_D + \xi_B}{\xi_B \xi_D + \xi_B + \xi_D}$, ${\rm{P}}_3$ can be computed as
\begin{align}
{\rm{P}}_3 & = \int^{+\infty}_{\frac{\xi_B \sigma_D^2}{p_C (\alpha-\xi_B+\alpha \xi_B)}} \frac{1}{\varphi_{CD}} e^{-\frac{h_{CD}}{\varphi_{CD}}} {\rm{d}} h_{CD} \nonumber \\
& = \exp \left[ -\frac{\xi_B \sigma_D^2}{\varphi_{CD} p_C (\alpha-\xi_B+\alpha \xi_B)} \right].
\end{align}
Otherwise, we have
\begin{align}
{\rm{P}}_3 & = \int^{+\infty}_{\frac{\xi_D \sigma_D^2}{p_C (1 - \alpha)}} \frac{1}{\varphi_{CD}} e^{-\frac{h_{CD}}{\varphi_{CD}}} {\rm{d}} h_{CD} \nonumber \\
& = \exp \left[ -\frac{\xi_D \sigma_D^2}{\varphi_{CD} p_C (1 - \alpha)} \right]
\end{align}
for $\frac{\xi_B \xi_D + \xi_B}{\xi_B \xi_D + \xi_B + \xi_D} < \alpha \leq 1$. The proof of Theorem \ref{theorem_2} ends here.

\section{Proof of Lemma \ref{lemma_2}}\label{appendix_a}
Depending on the relation between $\bar{p}_C$ and $P_C$, Lemma 1 can be proved by separately proving (\ref{constraint_of_case_1}) and (\ref{constraint_of_case_2}).
\subsubsection{Proof of (\ref{constraint_of_case_1})}
We use contradiction to prove (\ref{constraint_of_case_1}). Assuming ${R_{SC}}\left( {{\alpha ^*},p_C^*} \right) > {R_{CD,S}}\left( {{\alpha ^*},p_C^*} \right)$ for $\bar{p}_C < P_C$, and then there must exists a small enough $0 < \Delta {p_C} < P_C - \bar{p}_C$ which satisfies ${R_{SC}}\left( {{\alpha ^*},p_C^* + \Delta {p_C}} \right) > {R_{CD,S}}\left( {{\alpha ^*},p_C^* + \Delta {p_C}} \right)$. Hence, we have $R_D(\alpha^*, p_C^*)=R_{CD,S}(\alpha^*, p_C^*)$. Since ${R_B}\left( {\alpha ,p_C^{}} \right)$ and ${R_{CD,S}}\left( {\alpha ,p_C^{}} \right)$ are increasing functions of ${p_C}$ for a given $\alpha$, we have ${R_B}\left( {{\alpha ^*},p_C^*} \right) < {R_B}\left( {{\alpha ^*},p_C^* + \Delta {p_C}} \right)$ and ${R_{CD,S}}\left( {{\alpha ^*},p_C^*} \right) < {R_{CD,S}}\left( {{\alpha ^*},p_C^* + \Delta {p_C}} \right)$, which suggests that ${R_{\min}}\left( {{\alpha ^*},p_C^* + \Delta {p_C}} \right) > {R_{min}}\left( {{\alpha ^*},p_C^*} \right)$ and contradicts with the original assumption of the optimality of $\left( {{\alpha ^*},p_C^*} \right)$. Therefore, (\ref{constraint_of_case_1}) is proved.
\subsubsection{Proof of (\ref{constraint_of_case_2})}
Similar to the proof of (\ref{constraint_of_case_1}), we first assume ${R_{SC}}\left( {{\alpha ^*},p_C^*} \right) < {R_{CD,S}}\left( {{\alpha ^*},p_C^*} \right)$ for ${\bar p_C} \ge {P_C}$. Then after some algebraic deduction, we know that the optimal transmit power ${p_C}$ must satisfies $p_C^* > {\bar p_C}$, which is in the infeasible field of $\mathcal{OP}2$. Therefore, (\ref{constraint_of_case_2}) is proved and the proof of Lemma \ref{lemma_2} is completed.

\footnotesize{
\bibliographystyle{ieeetran}
\bibliography{bibtex}
}

\end{document}